\documentclass[twocolumn, 10pt]{IEEEtran}

\usepackage{floatflt}
\usepackage{color}
\usepackage{cite}
\usepackage{graphicx}
\usepackage{amsmath}
\usepackage{amsthm}
\usepackage{amssymb}
\usepackage{verbatim}
\usepackage{psfrag,array}
\usepackage{subfigure}
\usepackage{multirow}
\usepackage{hhline}
\usepackage{stfloats}
\usepackage{amsmath}
\usepackage{epstopdf}
\usepackage{algorithm}
\usepackage{algpseudocode}
\usepackage{multirow}
\usepackage{bigints}

\usepackage{slashbox}

\newcommand{\p}[1]{\mathop{\mbox{\it p} } }

\renewcommand{\vec}[1]{\ensuremath{\boldsymbol{#1}}}
\newcommand{\be}{\begin{equation}}
\newcommand{\ee}{\end{equation}}
\newcommand{\ba}{\begin{array}}
\newcommand{\ea}{\end{array}}
\newcommand{\bea}{\begin{eqnarray}}
\newcommand{\eea}{\end{eqnarray}}
\newcommand{\bean}{\begin{eqnarray*}}
\newcommand{\eean}{\end{eqnarray*}}

\newcommand{\rmh}{^{\dag}}

\definecolor{white}{rgb}{1,1,1}

\newtheorem{theorem}{Theorem}
\newtheorem{lemma}{Lemma}

\newtheorem{property}{Property}

\newtheorem{corollary}{Corollary}

\usepackage{xparse}

\ExplSyntaxOn
\NewDocumentCommand{\MeijerG}{smmmm}
 {
  \IfBooleanTF{#1}
   {
    \vic_meijerg:nnnnnn { #2 } { #3 } { #4 } { #5 } { small } { }
   }
   {
    \vic_meijerg:nnnnnn { #2 } { #3 } { #4 } { #5 } { } { \; }
   }
 }

\seq_new:N \l__vic_meijerg_args_in_seq
\seq_new:N \l__vic_meijerg_args_out_seq

\cs_new_protected:Nn \vic_meijerg:nnnnnn
 {
  \seq_set_split:Nnn \l__vic_meijerg_args_in_seq { | } { #3 }
  \seq_clear:N \l__vic_meijerg_args_out_seq  
  \seq_map_inline:Nn \l__vic_meijerg_args_in_seq
   {
    \seq_put_right:Nn \l__vic_meijerg_args_out_seq
     {
      \begin{#5matrix} ##1 \end{#5matrix}
     }
   }
  G\sp{#1}\sb{#2}
  \!\left(
  \seq_use:Nn \l__vic_meijerg_args_out_seq { #6\middle|#6 }
  #6\middle|#6
  #4
  \right)
 }
\ExplSyntaxOff

\begin{document}

\title{On Ergodic Capacity and Optimal Number of Tiers in UAV-Assisted Communication Systems}
\author
{
Sha Hu, \textit{Member, IEEE} 
\thanks{The author is with the Department of Electrical and Information Technology, Lund University, Lund, Sweden (email: sha.hu@eit.lth.se). }

}

\maketitle

\begin{abstract}
In this paper, we consider unmanned aerial vehicle (UAV) assisted communication systems where a number of UAVs are utilized as multi-tier relays between a number of users and a base-transceiver station (BTS). We model the wireless propagation channel between the users and the BTS as a Rayleigh product channel, which is a product of a series of independent and identically distributed (i.i.d.) Rayleigh multi-input multi-output (MIMO) channels. We put a special interested in optimizing the number of tiers in such UAV-assisted systems for a given total number of UAVs to maximize the ergodic capacity. To achieve this goal, in a first part we derive a lower-bound in closed-form for the ergodic capacity which is shown to be asymptotically tight as signal-to-noise ratio (SNR) increases. With the derived bound, in a second part we analyze the optimal number of UAV-tiers, and propose a low-complexity procedure that significantly reduces the search-size and yields near-optimal performance. Moreover, asymptotic properties both for the ergodic capacity of Rayleigh product channel, and the optimal solutions on number of tiers are extensively analyzed.

\end{abstract}

\begin{IEEEkeywords}
Unmanned aerial vehicle (UAV), Rayleigh product channel, ergodic capacity, upper-bound, lower-bound, signal-to-noise ratio (SNR), asymptotic properties, multi-tier relay, integer partition, optimization.
\end{IEEEkeywords}

\section{Introduction}
Unmanned aerial vehicles (UAVs) have gained much attention in advanced communication systems \cite{H18UAV, CG17, CE17, ZL16, MP15, BY16, LY18, MD16, MD18, CH17}. Due to the advantages such as mobility, flexibility, efficiency, and low-cost in deployment, cellular-connected UAVs have potentials in 5G beyond and IoT \cite{MD18, CB18, H18LIS, H18LIS1} systems. UAVs can be integrated into multi-tier relay networks as amplify-and-forward (AF) and decode-and-forward (DF) nodes to increase data-throughput and received signal-to-noise ration (SNR) \cite{CG79, WJ05, HA03, FB09, CK10}. This is particular helpful for sudden appearances of massive connections such as in a stadium or an outdoor concert,  and also for cellular users that are far away from a base-transceiver station (BTS) or obstructed by surrounding objects. UAVs can boost the connections in these circumstances by means of amplifying and beaming the signals.

Compared to a traditional terrestrial relay system, a UAV-assisted relay system is very flexible and has the capability to adapt its deployment according to real-time situations to maximize the performance. Further, the UAVs can appear anywhere anytime when there is an assignment, which makes it a powerful assistance to traditional cellular systems. By exploiting the flexibility in deployment, in this paper we consider multiple UAV-tier assisted communication systems in cellular networks. The target is to optimize the deployment for a given total number of UAVs through maximizing the ergodic capacity under different practical scenarios. 

Previous works on UAV assisted cellular networks can be referred to  e.g., \cite{BY16, LY18, MD16, MD18, CH17}. The potential and challenges of using UAVs in cellular networks were discussed in \cite{BY16, MD16, MD18, CH17}, as well as initial performance evaluation and trade-offs. The energy-efficiency and power control of UAVs were considered in \cite{WW18} and \cite{AY17}. The channel modeling and measurement of the air-to-ground (A2G) and air-to-air (A2A) communication channels were carried out in \cite{KD18, MS17, AJ14}. One observation from \cite{MS17, AJ14} is that both A2G and A2A channels will not always contain a line-of-sight (LoS) component. Especially in urban environment and for low-altitude UAVs, there are rich reflections and diffractions by surface-based obstacles such as tall-buildings, terrain, trees, and the UAV itself. For this fact and also for analytical tractability, in this paper we model the channels of considered UAV-assisted relay systems as independent and identically distributed (i.i.d.) Rayleigh channel. Although we do not consider other channel models, the analysis in this paper can be applied as a basis for studies on other channels. For instance, under the cases that there is LoS which yields Rician fading \cite{MS17}, the communication property obtained based only on the i.i.d. assumption can still apply. Moreover, under the cases that UAVs are deployed in rural areas or with high altitude, the LoS component becomes a dominant factor and an optimal deployment of UAVs is to minimize the distance between two adjacent tiers \cite{WB03, KA06} by using a single UAV at each tier.

By modeling the multi-input and multi-output (MIMO) channels between different UAV-tiers as i.i.d. Rayleigh fading, the effective propagation channel between the transmitting users and the receiving BTS is modeled as a Rayleigh product channel. The Rayleigh product channel origins from a double-scattering model \cite{GP02}, which comprises two i.i.d. channel components. Literatures addressing the achievable rate and diversity-multiplexing trade-off under such channels can be seen e.g.  in \cite{ZS04, JG08}. Latter, Rayleigh product channel is extended to comprise three channel components in \cite{FM10}, and then eventually to an ensemble of arbitrary $K$ i.i.d. Rayleigh MIMO channels as in \cite{YB07, AK13, RK14, F14}. Although eigenvalue statistics and ergodic capacity have been discussed in \cite{AK13, RK14, R02} for a Rayleigh product channel with $K$ tiers (in our case, $K\!-\!1$ UAV-tiers and one last BTS-tier), the results are based on hyper-geometric \textit{Meijer} $G$-function \cite{ET55} which is difficult to analyze. Capacity results for normal Rayleigh MIMO channels\footnote{By a normal Rayleigh MIMO channel, we refer to a direct MIMO channel between the users and the BTS without UAV, i.e., $K\!=\!1$.} can be reviewed as a special case with $K\!=\!1$ \cite{T99, AK13}. 

To optimize the number of tiers for a given total number\footnote{To simplify the description, we assume that both users and UAVs are equipped with a single-antenna. The cases that a user or a UAV is equipped with multiple antennas following similar analysis by treating each antenna as a separate user or UAV, respectively.} of $M$ UAVs, in principle one needs to evaluate all the integer partition sets of $M$ to find a partition set that maximizes the ergodic capacity $\tilde{R}$. With the analytical-form of $\tilde{R}$ in \cite{AK13}, calculating it for all partition sets requires extensive numerical computations or look-up-table operations, which renders a high cost and processing latency in real-time applications\footnote{As what becomes clear latter, the optimal number of UAV-tiers changes under different settings such as the numbers of antennas of the users and the BTS, the transmitting power, and the power attenuation factor. Therefore, the UAVs may need the capability to adapt to different practical scenarios.}. To simply the expression of $\tilde{R}$, one direct approach is to approximate it with an upper-bound through Jessen\rq{}s inequality. Such an obtained upper-bound can be tight when dimensions of the MIMO channel are sufficiently large, such as with traditional massive MIMO systems~\cite{T99, HT04}. However, with Rayleigh product channel the upper-bound becomes loose, due to the fact that the approximation errors of the upper-bound increases when the total number of tiers (i.e., the number of component random matrices in the Rayleigh product channel) increases. Therefore, finding other tight bound of $\tilde{R}$ is of interest.

For a given setting of UAV-tiers, the UAV-based relay system is similar to a traditional terrestrial relay system with the same settings. However, to our best knowledge, there is little work on considering optimizing the ergodic capacity $\tilde{R}$ under Rayleigh product channels. Previous works considering approximating and asymptotic properties of $\tilde{R}$ in multi-tier terrestrial wireless relay systems can be found in e.g., \cite{FB09, HA03,CK10, LS12, LL10, M14, NL11, S04}. But these works either consider the case that each tier has only a single-antenna relay \cite{FB09, HA03,CK10}, or there is only a single intermediate tier (the case when $K\!=\!2$) \cite{LL10}. Therefore, our analysis on ergodic capacity in the first part is also meaningful for traditional relay systems. There are also works consider the optimal of number of tiers in multi-tier terrestrial relay systems from different perspective. The authors in \cite{LS12} and \cite{M14} consider optimal power allocation and relay placements for multi-tier systems. In \cite{NL11}, the authors consider the optimal number of hops in a linear multi-tier AF relay model with maximizing a random coding error exponent (RCEE) instead of achievable rates. However, the expression of RCEE is also complex which make a direct optimization difficult. In \cite{S04}, the authors consider the optimum number of hops with time division multiple access (TDMA) multi-tier transmissions, which is optimized to minimize the transmission power for a given end-to-end rate.

In \cite{H18UAV}, we have derived a tight lower-bound of $\tilde{R}$ for the product of two Rayleigh MIMO channels, i.e., a single UAV-tier assisted communication system. We analyze trade-offs between the number of antennas and the transmit power of the UAV-tire in order to have higher ergodic capacity than a direct connection between the uses and the BTS. Following \cite{H18UAV}, we consider tight approximation of the ergodic capacity for multi-tier UAV assisted communication systems, and the optimization of number of tiers that maximizes $\tilde{R}$ for a given $M$ UAVs. We point out that we only consider the cases that the distance between adjacent UAV-tiers are relatively far and the channel can be model as Rayleigh fading such as in \cite{FB09, HA03,CK10, LS12, LL10, M14, NL11, S04}. Further, although we consider uplink transmission from users to the BTS, the analysis also applies to downlink transmission due to the channel reciprocity.

Although TDMA transmission can be used to mitigate cross-talks among UAV-tiers, one drawback is that $\tilde{R}$ is linearly scaled down by $K$. This renders the outcome that as SNR increases, the optimal number of tiers quickly decreases to 1 \cite{S04, OS06}. In our considered UAV-assisted system, we assume that the communications among UAV-tiers use approaches such as frequency-division multiplexing access (FDMA) \cite{MD16} (i.e., different tiers transmit on different frequency bands) or code-division multiplexing access (CDMA) \cite{LD13} (i.e., different tiers use orthogonal codes to spread transmit data). Further, with a pipelined transmission scheme on top of that, the number of tiers $K$ has negligible impact on the ergodic capacity \cite{FM10, NK04}, at a cost of wider bandwidth which can use free WIFI frequency band such as at 2.4 GHz or other bandwidth dedicated  for UAV communications.

Assuming there are $N_0$ users are connecting to a BTS with $N_K$ receiving antennas through $M$ UAVs, which are know aforehand. There are many integer partition sets of $M$ with
\bea \label{partM} \sum\limits_{k=1}^{K-1}N_k\!=\!M.\eea
With each partition scheme in (\ref{partM}) and together with $N_0$ and $N_K$, we form a UAV-assisted communication system with $K$ tiers according to a parameter setting ($N_0,\,N_1,\,\cdots,\,N_K$), where $N_k$ denotes the number of UAVs at the $k$th tier. Since the spatial multiplexing gain is determined by the minimum value of $N_k$ ($0\!\leq\!k\!\leq\!K$), it is not always optimal to put all UAVs in a single UAV-tier, i.e., setting $K\!=\!2$. On the other hand, there can be multiple schemes in (\ref{partM}) that have the same spatial multiplexing gain, but render different gains in terms of power attenuation and information-rate. With the derived lower-bound, these trade-offs can be directly evaluated and based on which, the number of tiers can be optimized to maximize the ergodic capacity $\tilde{R}$.

The main contribution of this paper are as follows:
\begin{itemize}
\item We derive a lower-bound of the ergodic capacity $\tilde{R}$ for Rayleigh product channel that comprises arbitrary $K$ i.i.d. rectangular Rayleigh MIMO channels with arbitrary dimensions. We show that the lower-bound is asymptotically tight as SNR increases and has a much simpler closed-form than its original form.
\item We show that the approximation error $\Delta\epsilon$ of the trivial upper-bound by switching the order of expectation operation and \lq\lq{}$\ln\!\det$\rq\rq{} function, asymptotically satisfies
\bea \label{eq2} \Delta\epsilon>N_0\sum_{k=1}^{K}\frac{1}{2N_k},\eea
which increases when $K$ increases, and $N_0$ is the minimum value of $N_k$ ($0\!\leq\!k\!\leq\!K$)
\item We show the differences between different settings of Rayleigh product channels such as with rectangular or square MIMO components. As a special case, adding an extra antenna\footnote{Note that, this result is only for Rayleigh product channel. For UAV-assisted cases, the impact on the received SNR also needs to be considered.} to the $k$th tier whose original number of antennas is $N_k$ can bring an increment to the ergodic capacity as
\bea \label{eq3}  \Delta\tilde{R}=\sum_{r=1}^{N_0}\frac{1}{N_k-\ell+1}.\eea
\item We analyze the optimal number of tiers for a given total $M$ UAVs with the derived bounds of $\tilde{R}$, and we show that the lower-bound based optimization is close-to-optimal and has much less computational-cost. Further, we propose an effective algorithm that significantly reduces the size of searching sets in the procedure, where we show that in general the optimal number of UAV-tier $K$ is
\bea \label{Ksub} K=\max\left(1+\biggl\lfloor \frac{M}{\min\{N_0, N_K\}}\biggr\rfloor,\;2\right).   \eea
\item We also analyze the asymptotic properties of the solutions and show that under low and high SNR cases, using a single UAV-tier ($K\!=\!2$) and setting each tier with a single UAV ($N_k\!=\!1$,$0\!<\!k\!<\!K$) is optimal for the two extreme cases, respectively.
\end{itemize}

The organization of the paper is as follows. In Section II, we briefly introduce the Rayleigh product channel model and the ergodic capacity. We also show a communication property between the tiers, and the high SNR property in Theorem 1. In Section III, we derive upper and lower bounds of the ergodic capacity and analyze the differences between them. We show the asymptotic properties of the lower-bound, and compare the ergodic capacity differences for different parameter settings of the Rayleigh product channel. In Section VI, we consider the number of tier optimization in the UAV-assisted systems, and propose a low-complexity algorithm which is shown to be effective. Simulation results are presented in Section V, and Section VI summarizes the paper.

\subsubsection*{Notation}
Throughout the paper, a capital bold letter such as $\vec{A}$ represents a matrix, a lower case bold letter $\vec{a}$ represents a vector, and matrix $\vec{I}$ represents an identity matrix. The superscripts $(\cdot)^\dag$ denotes the conjugate transpose of a matrix, and $(\cdot)^{-1}$ is the inverse. Further, $\ln(\cdot)$ is the natural logarithm function, $\det(\cdot)$ is the determinant, $\lfloor \cdot\rfloor$ and $\mod(\cdot)$ denotes the floor and modulo operations, respectively. In addition,  $\mathbb{E[\cdot]}$ is the expectation operator, $\mathrm{Tr(\cdot)}$ takes the trace of a matrix, and $\min(\cdot)$ takes the minimum of inputs.

\begin{figure*}[t]
\vspace*{-25mm}
\begin{center}
\hspace*{-30mm}
\scalebox{0.36}{\includegraphics{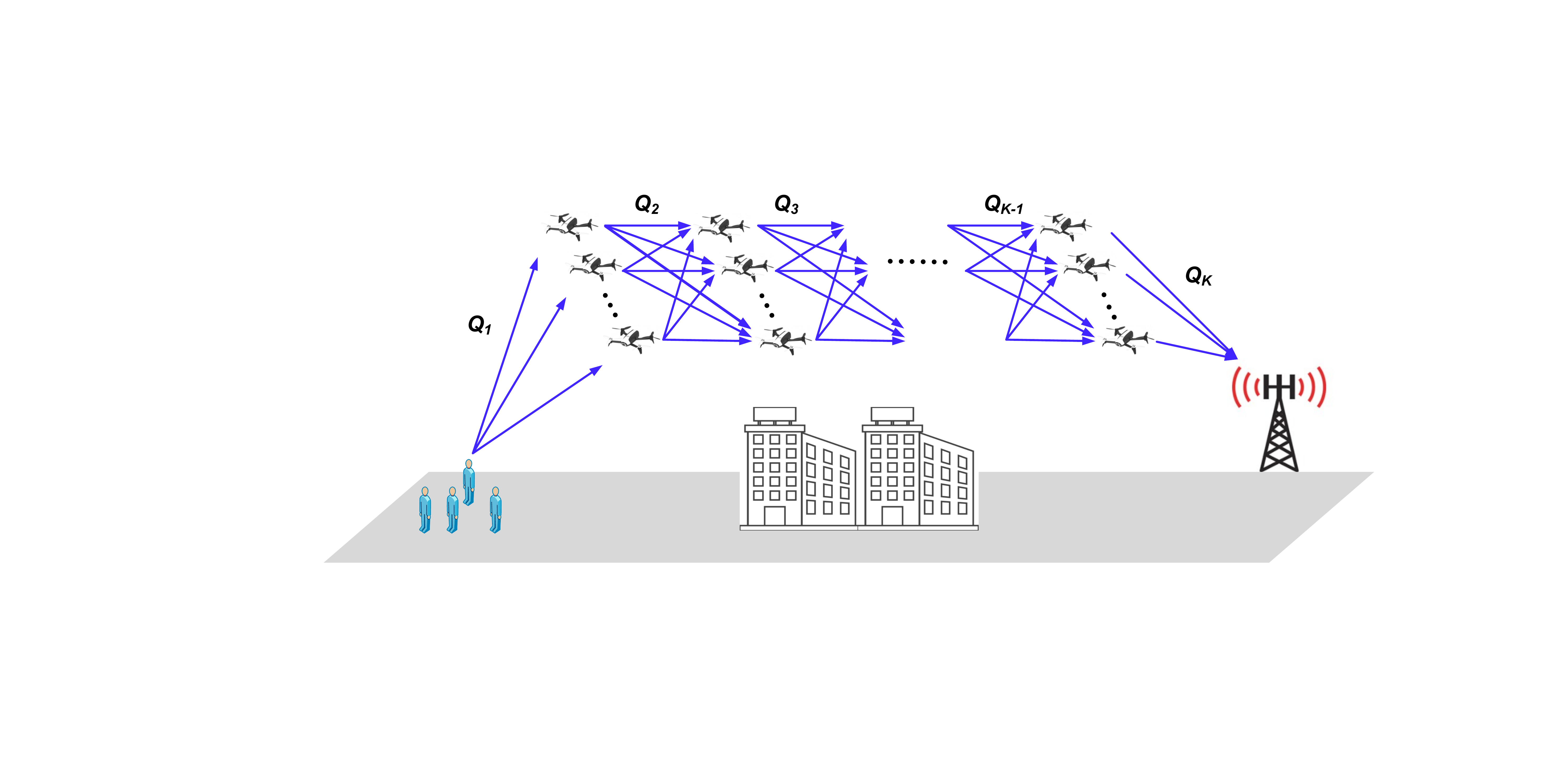}}
\vspace*{-34mm}
\caption{\label{fig1}Rayleigh product channel model in a UAV-assisted communication system with $K$ tiers, where the first $K\!-\!1$ tiers are the UAVs and the last tier is the BTS. The channels $\vec{Q}_k$ between adjacent tiers are modeled as i.i.d. Rayleigh MIMO channels, multiplying with the factors of power attenuations.}
\vspace*{-6mm}
\end{center}
\end{figure*}

\section{Preliminaries}

\subsection{Rayleigh Product Channel}

Consider a MIMO received signal model
\bea \label{md1} \vec{y}=\sqrt{q}\vec{H}{x}+\vec{n}, \eea
where $\vec{x}$ is the transmitted symbols from one or multiple users, and the transmit power\footnote{Although we call $q$ the transmit power, it however, denotes the combined impact of the transmit power, the power-amplifying in all UAV tiers, and the propagation losses.} is denoted as $q$. The channel $\vec{H}$ is of size $N_K\!\times\!N_0$, where $N_K$ denotes the number of receive antennas at a BTS, and $N_0$ is the total number of transmit antennas for users transmitting to the BTS simultaneously. For simplicity, we model $\vec{n}$ as additive Gaussian white noise (AWGN) with zero-mean and unit-variance. 

In a UAV-assisted communication system such as depicted in Fig. 1, the channel $\vec{H}$ is modeled as a Rayleigh product model. That is, each tier of the UAVs as depicted is assumed to be an independent scatter that beams the received signal from the previous tier (and with possible power-amplifying) to the next tier until it reaches the BTS. In other words, we assume
\bea \label{Hmd} \vec{H}=\vec{Q}_K\times\vec{Q}_{K-1}\times\cdots\times\vec{Q}_1=\prod\limits_{k=1}^K\vec{Q}_k, \eea
where $\vec{Q}_k$ are of size $N_k\!\times\!N_{k-1}$, and comprise of i.i.d. complex-valued Gaussian elements with zero-mean and unit-variance. Therefore, a parameter setting $(N_0,\,N_1,\,\cdots,\,N_K)$ uniquely determines the structure of $\vec{H}$. Note that in the rest of the paper, when we use the term $\prod\limits_{k=1}^K\vec{Q}_k$, it is always refereed to the multiplexing order in (\ref{Hmd}).  Further, in order to model the UAV-assisted communication systems with the received signal model (\ref{md1}), we assume that the received noise is white, which can be due to the fact that the noise power at each UAV is negligible compared to the received signal.

\subsection{A Communication Property of the Ergodic Capacity}

The capacity (nats per channel use) corresponding to the received signal model (\ref{md1}) equals
{\setlength\arraycolsep{1pt} \bea \label{R} R_{\left(\vec{Q}_1, \vec{Q}_2, \dots,\vec{Q}_K\right)}&=&\ln\!\det\!\left(\vec{I}\!+\!q\vec{H}\rmh\vec{H}\right) \notag \\
&=&\ln\!\det\!\left(\!\vec{I}\!+\!q\left(\prod\limits_{k=1}^K\vec{Q}_k\!\right)\rmh\!\left(\prod\limits_{k=1}^K\vec{Q}_k\right)\!\!\right)\!\!, \notag \\ \eea
\hspace{-1.2mm}and the ergodic capacity is
\bea \label{tR} \tilde{R}= \mathbb{E}\!\left[R_{\left(\vec{Q}_1, \vec{Q}_2, \dots,\vec{Q}_K\right)}\right]\!,\eea
where the expectation is taken over the probability density function (pdf) of $\vec{Q}_k$. Since different UAVs tier are independent from the others, we assume that \cite{G63}
\bea p\big(\vec{Q}_1, \vec{Q}_2, \dots,\vec{Q}_K\big)&=&\prod\limits_{k=1}^K p\big(\vec{Q}_k\big) \notag \\
&=&\prod\limits_{k=1}^K\frac{\exp\Big(\!\!-\!\text{Tr}\big\{Q_{k}\rmh Q_{k}\big\}\Big)}{\pi^{N_{K\!-\!1}N_K}}.\;\eea

We first state Property 1 that shows that permuting $N_k$ in a Rayleigh product channel will not change the ergodic capacity $\tilde{R}$, which is known as a weak commutation property for a product of i.i.d. random matrices in \cite{RK14, AK13}.
\begin{property}
The ergodic capacity $\tilde{R}$ of the Rayleigh product channel (\ref{md1}) is invariant under permutations of $(N_0, N_1, \cdots, N_K)$. 
\end{property}
\begin{proof}
See Appendix A.
\end{proof}

With Property 1, the analysis of $\tilde{R}$ for Rayleigh product channel is significantly simplified, as the order of $N_k$ is independent from the achieved ergodic capacity. With proper permutations we can always assume $N_0\!\leq\!N_1\!\leq\cdots\leq\!N_K$ when analyzing the properties of ergodic capacity. 

Letting $N_0\!=\!\min\limits_{0\leq k \leq K}\{N_k\}$, at high SNR\footnote{By high SNR we mean that either $q$ or the product $\prod\limits_{k=1}^K N_k$ is large, since the mean-value of the diagonal elements in $q\vec{H}\rmh\vec{H}$ equals $q\prod\limits_{k=1}^K N_k$. } it holds that
\bea \label{tRh} \tilde{R}&\approx&N_0\ln q+\mathbb{E}\!\left[\ln\det\!\left(\vec{H}\rmh\vec{H}\right)\right]\!.\eea

\begin{lemma}
For a Rayleigh product channel $\vec{H}$ in (\ref{Hmd}), it holds that
\bea \label{lemeq} \mathbb{E}\!\left[\ln\det\!\left(\!\vec{H}\rmh\vec{H}\right)\right]\!=\! \sum_{k=1}^{K}\mathbb{E}\!\left[\ln\det\!\left(\!\hat{\vec{Q}}_k\rmh\hat{\vec{Q}}_k\right)\!\right]\!, \eea
where $\hat{\vec{Q}}_k$ are i.i.d. random Rayleigh MIMO channels with dimensions $N_k\!\times\!N_0$.
\end{lemma}
\begin{proof}
See Appendix B.
\end{proof}

From Lemma 1, $\tilde{R}$ can be expressed as a summation over i.i.d. Rayleigh MIMO channels $\hat{\vec{Q}}_k$ at high SNR, but with reduced dimensions $N_k\!\times\!N_0$, instead of the original $N_k\!\times\!N_{k-1}$ of $\vec{Q}_k$. Since $\hat{\vec{Q}}_k\rmh\hat{\vec{Q}}_k$ is complex Wishart distributed, it can be readily seen from \cite{OP02, G63} that
 \bea \label{EHH} \mathbb{E}\!\left[\ln\det(\hat{\vec{Q}}_k\rmh\hat{\vec{Q}}_k)\right] &=&\sum_{\ell=1}^{N_0}\psi(N_k-\ell+1) \notag \\
 &=&-N_0\gamma+\sum_{\ell=1}^{N_0}\sum_{r=1}^{N_k-\ell}\frac{1}{r}, \eea
where the \textit{digamma} function $\psi(n)$ is $$\psi(n)=-\gamma+\sum_{k=1}^{n-1}\frac{1}{k}$$
and $\gamma\!\approx\!0.5772$ is the \textit{Euler-Mascheroni} constant.

Inserting (\ref{EHH}) back into (\ref{lemeq}), we have the below Theorem 1, which will be useful in deriving a lower-bound for $\tilde{R}$.

\begin{theorem}
For a Rayleigh product channel $\vec{H}$, it holds that
 \bea \label{c2} \mathbb{E}\!\left[\ln\det(\vec{H}\vec{H}\rmh)\right] &=&\sum_{k=1}^{K}\sum_{\ell=1}^{N_{0}}\psi(N_k-\ell+1) \notag \\
 &=&-KN_0\gamma+\sum_{k=1}^{K}\sum_{\ell=1}^{N_0}\sum_{r=1}^{N_k-\ell}\frac{1}{r}. \eea
\end{theorem}

From Theorem 1, we see that $N_0$ (the minimum of $N_k$) and $K$ (the total number of tiers) play fundamental roles in the ergodic capacity that can be achieved for the Rayleigh product channel $\vec{H}$.

\section{Bounds of the Ergodic Capacity for Rayleigh Product Channel}

\subsection{Exact-Form of the Ergodic Capacity}
Given the Rayleigh product channel model (\ref{md1}), the ergodic capacity $\tilde{R}$ in (\ref{tR}) can be solved in an analytical-form stated in Lemma 2 \cite{AK13}.
\begin{lemma}
With Rayleigh product channel model (\ref{md1}), the ergodic capacity $\tilde{R}$ (nats per channel use) in (\ref{tR}) equals
\bea \tilde{R}&=& \!\sum_{n=0}^{N_0-1}\sum_{m=0}^{n}\!\Bigg(\sum_{s=0}^{n}\frac{(-1)^{m+s}n!(n+\nu_1)!}{(n-m)!m!(n-s)!s!(s+\nu_1)!} \notag \\
&& \quad\; \times\, \MeijerG*{K+2,1}{2,K+2}{0,\,1 \\ m+1+\nu_K,\,\cdots,\,m+1+\nu_2,\,m+s+1+\nu_1,\,0,\,0 } {q^{-1}}\!\!\Bigg) \notag \\ 
&& \quad\; \times\Bigg(\prod_{k=1}^{K}\frac{1}{(m+\nu_k)!}\Bigg),
\eea
where the positive dimension differences
$$ \nu_k=N_k-N_0.$$
\end{lemma}

Although $\tilde{R}$ can be expressed in analytical-form, it is complex to evaluate with involving the \textit{Meijer} $G$-function and the \textit{gmama} function $\Gamma(\cdot)$, and the \textit{Meijer} $G$-function is defined as a line integral on the complex-plane as \cite{ET55}
\bea \label{exacttR} &&\MeijerG*{m, n}{p, q}{a_1,\,a_2,\,\cdots,\,a_p \\ b_1,\,b_2,\,\cdots,\,b_q } {z}\notag \\
&& =\frac{1}{2\pi i}\bigintss_{L}\frac{\prod\limits_{j=1}^m\Gamma(b_j-s)\prod\limits_{j=1}^n\Gamma(1-a_j+s)}{\prod\limits_{j=m+1}^q\Gamma(1-b_j+s)\prod\limits_{j=n+1}^p\Gamma(a_j-s)}z^s \mathrm{d}z. \notag \eea
Further, it is also difficult to understand the connections between different parameter settings $(N_0, N_1, \cdots, N_K)$ and the attained ergodic capacity $\tilde{R}$. Hence, next we find bounds for $\tilde{R}$ with simpler forms.

\subsection{Upper and Lower Bounds}

By Jessen\rq{}s inequality, the ergodic capacity $\tilde{R}$ can be upper bounded as
{\setlength\arraycolsep{2pt} \bea \label{ub}\tilde{R}&\leq&\ln\!\det\!\left(\vec{I}+q\mathbb{E}\big[\vec{H}\rmh\vec{H}\big]\right)\notag \\
&=&N_0\ln\!\left(\!1+q\prod\limits_{k=1}^{K}N_k\!\right)\!.\eea}
\hspace{-1.4mm}This bound is trivial and widely used to approximate the ergodic capacity for normal Rayleigh channels, such as in massive MIMO systems \cite{T99, HT04}. For Rayleigh product channel ($K\!\geq\!2$), however, this upper-bound becomes loose as what will be explained later. Therefore, seeking another bound that is tight is of interest. Following the similar idea in \cite{H18UAV}, we derive a lower-bound that is asymptotically tight, which is stated in Property 2 together with the aforementioned upper-bound (\ref{ub}).

\begin{property}
The ergodic capacity of the Rayleigh product channel model (\ref{md1}) is bounded as
 \be \label{tR3} N_0\ln\!\big(1+q\exp\Big(g-K\gamma\big)\Big)\leq \tilde{R}\leq N_0\ln\!\left(\!1+q\prod\limits_{k=1}^{K}N_k\!\right)\!, \ee
where  
\bea \label{gK} g&=&K\gamma+\frac{1}{N_0}\sum_{k=1}^{K}\sum_{\ell=1}^{N_0}\psi(N_k-\ell+1)\notag \\
&=& \frac{1}{N_0}\sum_{k=1}^{K}\sum_{\ell=1}^{N_0}\sum_{r=1}^{N_k-\ell}\frac{1}{r}. \eea
\end{property}
\begin{proof}
See Appendix C.
\end{proof}

\subsection{Asymptotic Properties of the Bounds}
Under cases that $q\exp\Big(g-K\gamma\big)\!\gg\!1$, it holds that
 \be   N_0\ln\!\big(1+q\exp\Big(g-K\gamma\big)\Big)\approx N_0\ln q+N_0\Big(g-K\gamma\big). \ee
Then, from (\ref{tRh}) and Theorem 1 we have the below corollary.
\begin{corollary}
The lower-bound in (\ref{tR3}) for the ergodic capacity $\tilde{R}$ is asymptotically tight.
\end{corollary}
 
To show the gap between the derived upper and lower bounds, we notice that the difference between them is asymptotically equal to $\Delta\epsilon$ and
\bea \frac{\Delta\epsilon}{N_0}=-\big(g-K\gamma\big)+ \sum\limits_{k=1}^{K}\ln N_k.  \eea
Using the approximation of \textit{digamma} function \cite{ET55} that
$$\psi(x)\approx\ln x-\frac{1}{2x}, \;\; x\!>\!1,$$
and by the definition of $g$, $\Delta\epsilon$ can be approximated as
\bea   \frac{\Delta\epsilon}{N_0}&=&-\frac{1}{N_0}\sum_{k=1}^{K}\sum_{\ell=1}^{N_0}\psi(N_k-\ell+1)+ \sum\limits_{k=1}^{K}\ln N_k\notag \\
&\approx& \frac{1}{N_0}\sum_{k=1}^{K}\sum_{\ell=1}^{N_0}\!\left(\!\ln\!\left(\!\frac{N_k}{N_k-\ell+1}\!\right)\!+\frac{1}{2(N_k-\ell+1)}\!\right)\!\!.\qquad \eea
Therefore, the ergodic capacity difference $\Delta\epsilon$ satisfies
\bea  \label{deltae}  \Delta\epsilon&>&\sum_{k=1}^{K}\sum_{\ell=1}^{N_0}\frac{1}{2(N_k-\ell+1)}   \\
&\geq&N_0\sum_{k=1}^{K}\frac{1}{2N_k}. \notag \eea

As the lower-bound is asymptotically tight, the ergodic capacity difference in (\ref{deltae}) is asymptotically equal to the errors between the upper-bound and the exact value of $\tilde{R}$. As can be seen from (\ref{deltae}), the error $\Delta\epsilon$ increases as $K$ increases, and in order for $\Delta\epsilon$ to be close to zero, it must holds that $N_k\!\gg\!N_0$ for $k\!>\!0$.

Using Corollary 1 we can obtain a below corollary for a normal Rayleigh MIMO channel (i.e., $K\!=\!1$). 
\begin{corollary}
The ergodic capacity $\tilde{R}$ when $q\!\to\!\infty$ for a normal Rayleigh channel $\vec{H}$ of sizes $N_0\!\times\!N_1$ ($N_1\!\geq\!N_0$) can be approximated as
\bea \label{ntR}\tilde{R}&=&\mathbb{E}\big[\ln\!\det\!\left(\vec{I}+q\vec{H}\rmh\vec{H}\right)\!\big]\notag \\
&\approx&  N_0\big(\ln\! q-\gamma\big)+\sum_{\ell=1}^{N_0}\sum_{r=1}^{N_1-\ell}\frac{1}{r}.  \eea
\end{corollary}
Note that when $N_1\!\gg\!N_0$, the harmonic series
$$\sum_{r=1}^{N_1}\frac{1}{r}\approx\ln\! N_1+\gamma,$$
and (\ref{ntR}) becomes
\bea \label{hsnrbd} \tilde{R}\approx N_0\ln\! \big(qN_1\big), \eea
which is aligned with the upper-bound in Property 2 for the case $K\!=\!1$. However, (\ref{hsnrbd}) only holds for cases $N_1\!\gg\!N_0$, but the derived (\ref{ntR}) holds for general settings of $N_0$ and $N_1$.

\subsection{The Connection between Rectangular and Square Random Matrices in Rayleigh Product Channel}

Based on Property 2, we have Property 3 that states the ergodic capacity difference between the Rayleigh product channel formed by a number of rectangular and square i.i.d. random matrices.

\begin{property}
At high SNR the ergodic capacity increment $\Delta\tilde{R}$, between Rayleigh product channels (for an identical $q$) with a parameter setting $(N_0,\, N_1,\, \cdots,\, N_K)$ and with square matrices $N_k\!=\!N_0$ ($1\!\leq\!k\!\leq\!K$), is
\bea \Delta\tilde{R}&=&\sum_{k=1}^{K}\sum_{\ell=1}^{N_{0}}\psi(N_k-\ell+1)-\sum_{k=1}^{K}\sum_{\ell=1}^{N_{0}}\psi(N_0-\ell+1) \notag \\ &=& \sum_{k=1}^{K}\sum_{\ell=1}^{N_{0}}\sum_{s=N_0}^{N_k-1}\frac{1}{s-\ell+1}. \notag \eea
\end{property}
\begin{proof}
See Appendix D.
\end{proof}

To interpret Property 3, as a special case, we consider adding an extra antenna\footnote{Such an operation changes both the dimensions of $\vec{Q}_{k-1}$ and $\vec{Q}_k$. However, if $N_k\!<\!N_{k+1}$, it still holds $\tilde{N}_k\!\leq\!N_{k+1}$. If $N_k\!=\!N_{k+1}$, then adding an extra-antenna to $N_k$ is equivalent to add that antenna to $N_{k+1}$ which yields the same capacity increment.} in the UAV-assisted communication system by increasing $\tilde{N}_k\!=\!N_k\!+\!1$ can bring an increment $\Delta\tilde{R}$ that is asymptotically equal to (\ref{eq3}), that is,
\bea   \Delta\tilde{R}=\sum_{r=1}^{N_0}\frac{1}{N_k-\ell+1}. \notag \eea 

\begin{figure*}[t]
\vspace*{-3mm}
\begin{center}
\hspace*{-0mm}
\scalebox{0.43}{\includegraphics{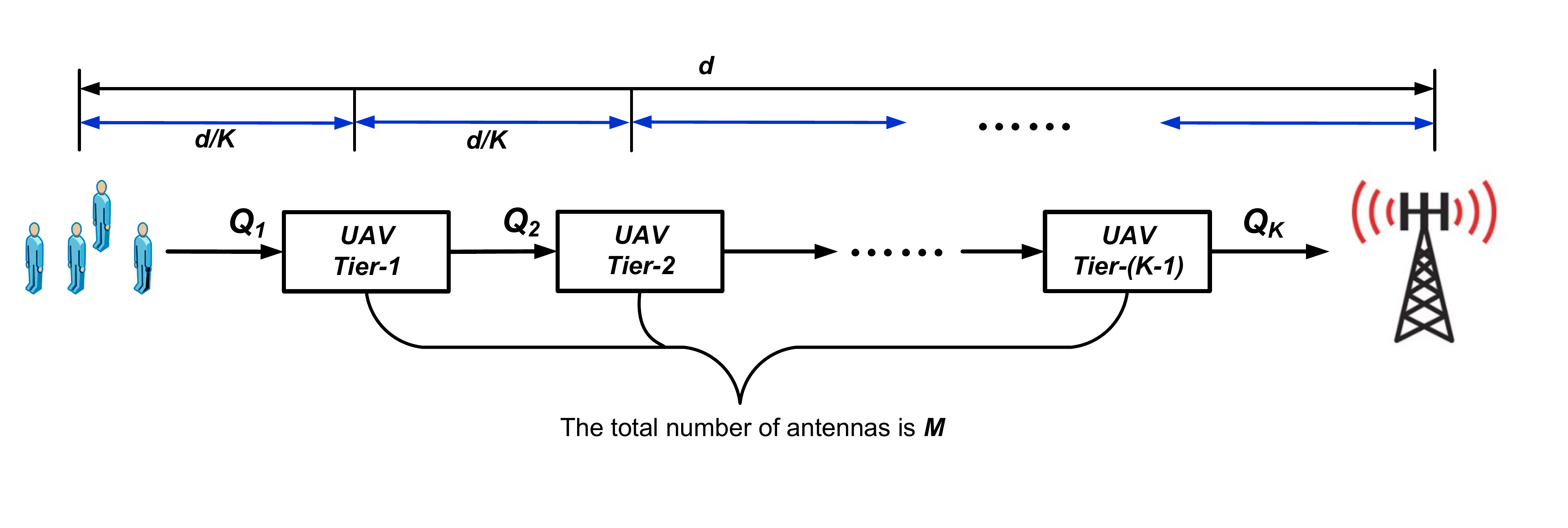}}
\vspace*{-12mm}
\caption{\label{fig2}Partitioning the total number of $M$ UAVs into $K\!-\!1$ tiers with parameters $(N_1,\,N_2,\,\cdots,\,N_{K-1})$, and the values of $N_0$ and $N_K$ are fixed.}
\vspace*{-6mm}
\end{center}
\end{figure*}

\section{Number of Ties Optimization}
With the derived bounds, in this section we consider optimizing the number of tiers in a UAV-assisted communication system with a given $M$ UAVs. Such a UAV relay system is analogous to a linear multi-tier terrestrial relay system \cite{NL11, S04, OS06}, but the target now is to find an optimal partition of $M$ that splits it into $K\!-\!1$ integers that satisfy (\ref{partM}) according to different scenarios, which yield a deployment of UAV-assisted relay system that can have the highest ergodic capacity for each applied scenario.

There are obvious trade-offs between the number of tiers $K$ and the number of UAVs at each tier $N_k$. When $K$ is larger, the minimum value of $N_k$ becomes smaller. That is, the spatial multiplexing gain is reduced seen from (\ref{tR3}). On the other hand, when $K$ is smaller, the distance between two adjacent tiers are larger. The power attenuation factor can be modeled as \cite{S04} 
\bea  \eta\propto\Big(\frac{d}{K}\Big)^{\!-\alpha}, \eea
where $d$ is the distance between the users to the BTS, and $\alpha$ is the path-loss exponent with typical values between 2 and 4.

Further, for fair comparisons we assume that the transmit-power at each antenna of the UAV is equal to $p$. That is, at the $k$th UAV-tier, the received signal at each antenna is scaled by a factor $p/N_{k-1}$, and the total transmit power of all UAVs is equal to $Mp$. Note that, the cases different UAVs with unequal transmit power follow the similar analysis since this will only impact the definition of parameter $q$ in the ergodic capacity of the modeled Rayleigh product channel.

With the above assumptions, the ergodic capacity for the UAV-assisted systems can be modeled as
\bea \label{tRK}\tilde{R}&=&\mathbb{E}\big[\ln\det\!\left(\vec{I}+q\vec{H}\rmh\vec{H}\right)\!\big],  \eea
where
\bea \label{qcon} q=cp_0\frac{K^\alpha p^{K-1}}{\prod\limits_{k=0}^{K-2}N_k},\eea
and $c$ is a constant representing the power attenuation from the users to the BTS with respect to the distance $d$, and $p_0$ is the transmit power from each of the users. Without loss of generality, we let
$$\tilde{p}=(cp_0)^{\frac{1}{K-1}}p,$$
and (\ref{qcon}) can be rewritten as
\bea \label{qcon1} q=\frac{K^\alpha \tilde{p}^{K-1}}{\prod\limits_{k=0}^{K-2}N_k}.\eea
With the modeled Rayleigh product channel $\vec{H}$ in (\ref{tRK}), the analysis on ergodic capacity in Sec. III can be used for optimizing the deployment of UAVs.

\subsection{Problem Formulation}

The optimization problem can be formulated as:
{\setlength\arraycolsep{2pt} \bea \label{prbm} &&\underset{\begin{subarray}{c}
K\\
 (N_1,N_2,\cdots,N_{K-1})  \end{subarray}}{\text{maximize}} \;\;  \tilde{R} \text{\,\;in\;} (\ref{tRK})\notag   \\
&&\quad\mathrm{subject\; to\;\;\;\;\;}  (\ref{partM}) \text{\;and\;} (\ref{qcon1}).\eea}
\hspace{-1.5mm}where both $K$ and $(N_1,N_2,\cdots,N_{K-1})$ are yet to be optimized, and $N_0$ and $N_K$ are known parameters in computing $\tilde{R}$ that denotes the number of antennas of the users and the BTS, respectively.

An exhaust search over all possible partition sets yields a prohibitive complexity when $M$ is large\footnote{Under certain circumstances, the optimization can be simplified. For instance, if the values $q$ in (\ref{tRK}) are unaltered for different settings such as the UAVs are only used as scatters, then to optimize the ergodic capacity with the upper-bound is equivalent to maximum the product of the elements in the partition set of $M$. The optimal partition follows the rule that $M$ are partitioned only with 2 and 3, and with as many 3\rq{}s as possible \cite{D05, B93}.}. An asymptotic expression of the number\footnote{We use $\#M$ to denote the number of integer partitions of $M$.} of integer partitions for $M$ is \cite{HR18}
\bea  \label{nbset} \#M\approx\frac{1}{4\sqrt{3}M}\exp\left(\! \pi\sqrt{\frac{2M}{3}}\right)\!,   \eea
which increases rapidly as $M$ increases. Due to the numerical calculations needed for evaluating hyper-geometric functions, directly solving (\ref{prbm}) with the exact-form of $\tilde{R}$ in Lemma 2 is also complex. Therefore, we consider to use the derived bounds.

If the upper-bound of $\tilde{R}$ in Property 2 is used in the optimization problem (\ref{prbm}), $\tilde{R}$ is approximated as
\bea  \label{optub} \tilde{R}\approx\tilde{N}_0\ln\!\left(\!1+\frac{K^\alpha  \tilde{p}^{K-1} N_{K-1}N_K}{N_0}\!\right)\!. \eea
Note that, $\tilde{N}_0$ denotes the minimal values among all $N_k$ including $N_0$ and $N_K$, and $N_0$ is the number of users which is fixed in the optimizations.

Similarly, if we use the lower-bound in (\ref{prbm}), $\tilde{R}$ can be expressed as
\bea  \label{optlb} \tilde{R}\approx\tilde{N}_0\ln\!\left(\!1+\frac{K^\alpha  \tilde{p}^{K-1}}{\prod\limits_{k=0}^{K-2}N_k}\exp\Big(g-K\gamma\big)\!\right)\!, \eea
where
$$g= \frac{1}{\tilde{N}_0}\sum_{k=1}^{K}\sum_{\ell=1}^{\tilde{N}_0}\sum_{r=1}^{N_k-\ell}\frac{1}{r}. $$

As can be seen, both optimizations with the expressions in (\ref{optub}) and (\ref{optlb}) need to search over all possible partition sets of $M$, despite that the optimization (\ref{optub}) is slightly simple since only $\tilde{N}_0$ and $N_{K-1}$ need to be considered.

\subsection{The Proposed Optimization Routine}

To further reduce the search-size in the optimizations, we notice that there are two main principles to maximize the ergodic capacity $\tilde{R}$ for a given $\tilde{p}$:
\begin{enumerate}
\item $\tilde{N}_0$, the minimum of all $N_k$, shall be maximized.
\item $K$, the total number of tiers, shall be maximized to reduce the power attenuations.
\end{enumerate}
With the above two principles, for a given pair $(N_0, N_K)$, the optimal value of $K$ can be determined via (\ref{Ksub}), that is,
\bea  K=\max\left(1+\biggl\lfloor \frac{M}{\min\{N_0, N_K\}}\biggr\rfloor,2\right)\!.  \notag \eea
Then, what left is to find all possible partitions sets for the remainder
$$R=M-(K-1)\min\{N_0, N_K\},$$
if $R\!>\!0$. That is, denoting $(r_1, r_2,\cdots,r_t)$ as a partition set of $R$, the corresponding partition set of $M$ is set to
 \bea \label{subpart} N_k=\left\{\begin{array}{cc}\min\{N_0, N_K\}\quad & \quad 1\leq k\leq K-t-1,\\
 \min\{N_0, N_K\}+r_t\quad & \quad K-t\leq k\leq K-1 . \end{array}\right. \quad \eea
In total only $\#R$ partition sets need to be evaluated, which yields great search-size reduction, due to the fact that 
$$\#R\!\ll\!\#M.$$
For instance, letting $\min\{N_0, N_K\}\!=\!3$ and $M\!=\!16$, the number of partition sets $\#M\!=\!231$, while $\#R\!=\!1$ as $R\!=\!1$. That is, the optimal solution is directly given as the partition set $\{5, 5, 6\}$ of $M$, which is also aligned with the numerical simulation result in Fig. 7 shown later Sec. V.

To further reduce the complexity, a further simplification is to let $r_1\!=\!R$, which gives a suboptimal solution of (\ref{prbm}) directly as
 \bea \label{subpart1} N_k=\left\{\begin{array}{cc}\min\{N_0, N_K\}\qquad & \qquad 1\leq k\leq K-2,\\
 \min\{N_0, N_K\}+R\qquad & \qquad k=K-1 . \end{array}\right.\eea
The idea is to maximize $q$ in (\ref{qcon}) for a given $K\!-\!1$, or equivalently, minimize the term $\prod\limits_{k=0}^{K-2}N_k$ by adding the reminder $R$ onto $N_{K-1}$.

As an example, in Table I we list the optimal $K$ for $M\!=\!20$ and different values of ($N_0$, $N_K$). One disadvantage of the suboptimal solution  (\ref{subpart}) is that the impact of SNR is not taken into account. However, as shown later by simulation results, the suboptimal solution (\ref{subpart}) is close-to-optimal in a wide range of SNR values. At extreme low or high SNR values, the optimal solutions can also be derived as shown in the next. Therefore, a practical optimizing approach is to combine both (\ref{subpart}) and the asymptotic solutions.

\begin{table}[ht!]
\renewcommand{\arraystretch}{1.5}
\vspace{-0mm}
\centering
\caption{Optimal $K$ in (\ref{Ksub}) with $M\!=\!20$ and different ($N_0$, $N_K$).}
\label{tab1}
\vspace{-1mm}
\begin{tabular}{|c|c|c|c|c|c|c|c|c|}
\hline
 \backslashbox{$N_0$}{$N_K$} & 8 &16& 32 & 48& 64 &96&128&256 \\ \hline
2&11&11 &11&11&11&11 &11&11\\ \hline
 4&6&6&6&6&6&6&6&6\\ \hline
  6&4&4&4&4&4&4&4&4\\ \hline
   8&3&3&3&3&3&3&3&3 \\ \hline
    10&3&3&3&3&3&3&3&3\\ \hline
     12&3&2 &2&2&2&2&2&2 \\ \hline
      14&3&2&2&2&2&2&2&2 \\ \hline
       16&3&2 &2&2&2&2&2&2\\ \hline
\end{tabular}
\vspace{-2mm}
\end{table}

\subsection{Asymptotic Solutions and Practical Optimization Procedure}}
Under the case that $p\!\gg\! K$ is sufficiently large, it holds from both (\ref{optub}) and (\ref{optlb}) that
$$ \tilde{R}\approx\tilde{N}_0(K-1)\ln q.$$
In such case, the optimal number of tiers can be optimized through maximizing $\tilde{N}_0(K-1)$. For a given $\tilde{N}_0$, the maximal value of $K\!-\!1$ is $\bigl\lfloor M/\tilde{N}_0\bigr\rfloor$, and hence,
$$\max(\tilde{N}_0(K-1))\leq M,$$
which can be achieved by a partition set\footnote{However, a partition set with all 1\rq{}s is not a unique solution to achieve the maximum. For instance, $\tilde{N}_0\!=\!K\!-\!1\!=\!4$ is also optimal for $M\!=\!16$.}  with all $N_k\!=\!1$. This is to say, when SNR increases, setting the number of tiers to $M$ with each UAV-tier only containing a single UAV is close to optimal.

On the other hand, under the case that $p$ is sufficiently small, $p^{K-1}$ decreases as $K$ increases, and the optimal number of tiers is $K\!=\!2$, that is, using a single UAV-tier comprising all $M$ UAVs is close-to-optimal.

Note that, these conclusions are different from the observations in \cite{S04, OS06}, due to the fact that they assume TDMA transmission schemes and the ergodic capacity $\tilde{R}$ linearly decreases in $K$. However, similar optimizations for the number of tiers with TDMA transmissions can follow the same analysis shown above.

Combing the above discussions, a practical optimization procedure for the number of UAV-tier with low-complexity is to find the partition set that maximizes the lower-bound derived in (\ref{optlb}), and with the partition sets defined in (\ref{subpart}) and the two asymptotic settings for low and high SNR cases. Such an optimizing approach yields a significantly reduced search-size of ($\#R+2$), and is more robust against SNR changes.

\section{Numerical Results}

In this section, we show simulations results with the consider UAV-assisted communication systems and the Rayleigh product channels. We use various settings such that the previous elaborated properties can be clearly explained.

\subsection{Tightness of the Lower-bound}
In Fig. 3 and 4, we show comparisons between the derived bounds and the numerical results of the ergodic capacity $\tilde{R}$. In both cases we test with a UAV-assisted system with three tiers, i.e., $K\!=\!3$. In Fig. 3, we set $N_0\!=\!N_1\!=\!N_2\!=\!4$, and $N_3\!=\!8$, while in Fig. 4 we set $N_0\!=\!4$, $N_1\!=\!N_2\!=\!8$, and $N_3\!=\!16$, respectively. As can be seen, in both cases, the derived lower-bounds are much tighter than the traditional upper-bounds. Further, as $q$ increases, the lower-bounds become tight and converge to the exact $\tilde{R}$. Moreover, with larger values of $N_k$ such as in Fig. 4, the lower-bound is also tight even with small values of $q$. These results are well aligned with the derivations in Sec. III-B.

\subsection{Asymptotic Properties}
In Fig. 5, we show asymptotic properties of $\tilde{R}$ with numerical simulations. We evaluate $\tilde{R}$ for three different scenarios, but all with $K\!=\!4$, $N_0\!=\!3$, and $N_4\!=\!8$. In the first case, we set $N_1\!=\!N_2\!=\!4$, while in the second case we only increase $N_2\!=\!5$ and the others remain unchanged. According to Property~3, adding one extra antenna the increment of $\tilde{R}$ under high SNR equals
$$  \sum_{r=1}^{3}\frac{1}{r+1}\approx 1.08. $$
In the third case, we further increase both $N_2\!=\!5$ and $N_3\!=\!6$, the increment of $\tilde{R}$ over the first case under high SNR according to Property~3 now is
$$ 2\sum_{r=1}^{3}\frac{1}{r+1}+\sum_{r=1}^{3}\frac{1}{r+2}\approx 2.95. $$
As can be seen, these two values are well aligned with the numerical results shown in the lower part of Fig. 5, where we use the ergodic capacity of the latter two cases and subtract them from the first case, respectively.

\begin{figure}[t]
\vspace*{-3mm}
\begin{center}
\hspace*{-5mm}
\scalebox{0.42}{\includegraphics{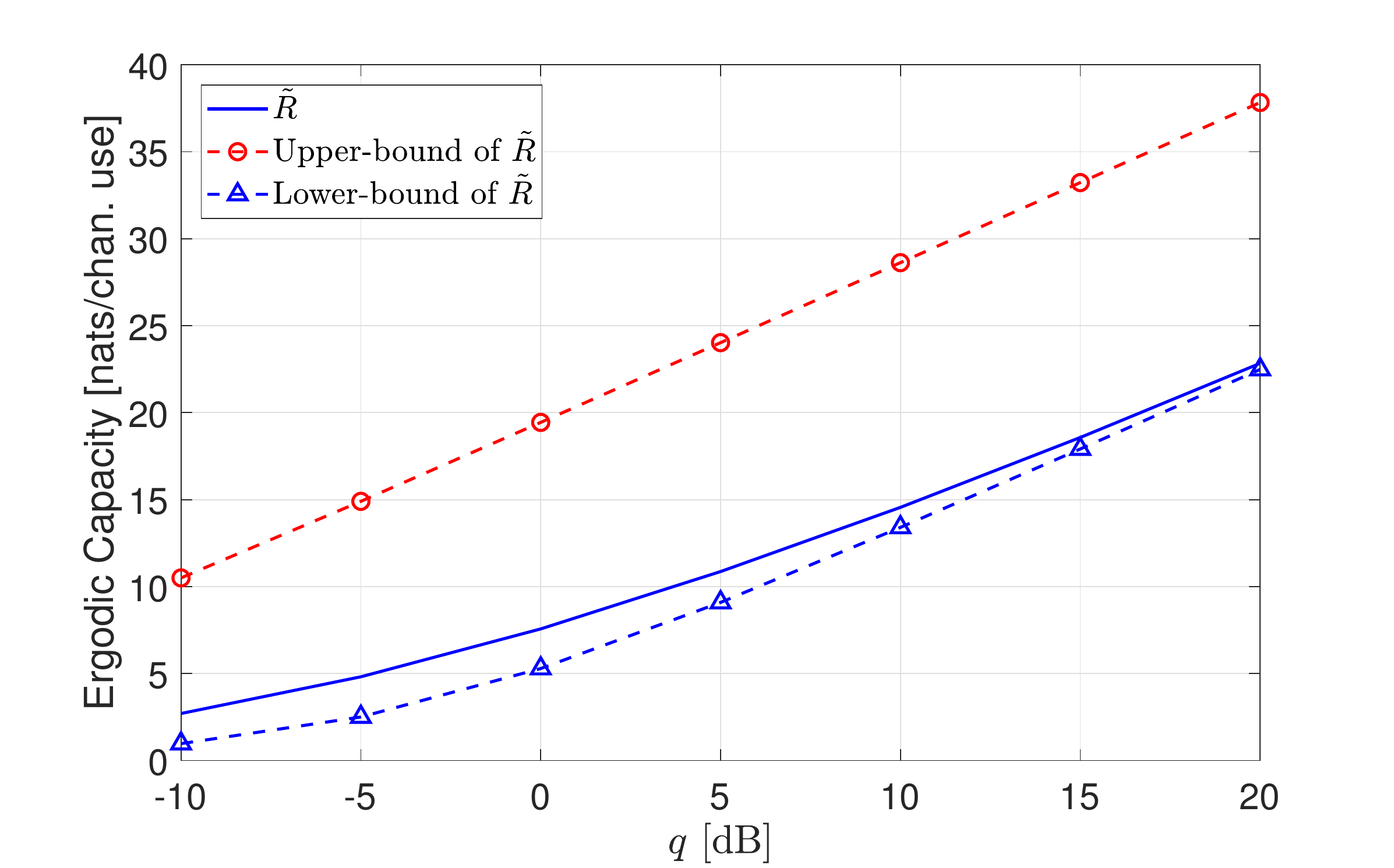}}
\vspace*{-8mm}
\caption{\label{fig3}The ergodic capacity with $K\!=\!3$, $N_0\!=\!N_1\!=\!N_2\!=\!4$, and $N_3\!=\!8$.}
\vspace*{-4mm}
\end{center}
\end{figure}

\begin{figure}
\vspace*{-2mm}
\begin{center}
\hspace*{-5mm}
\scalebox{0.42}{\includegraphics{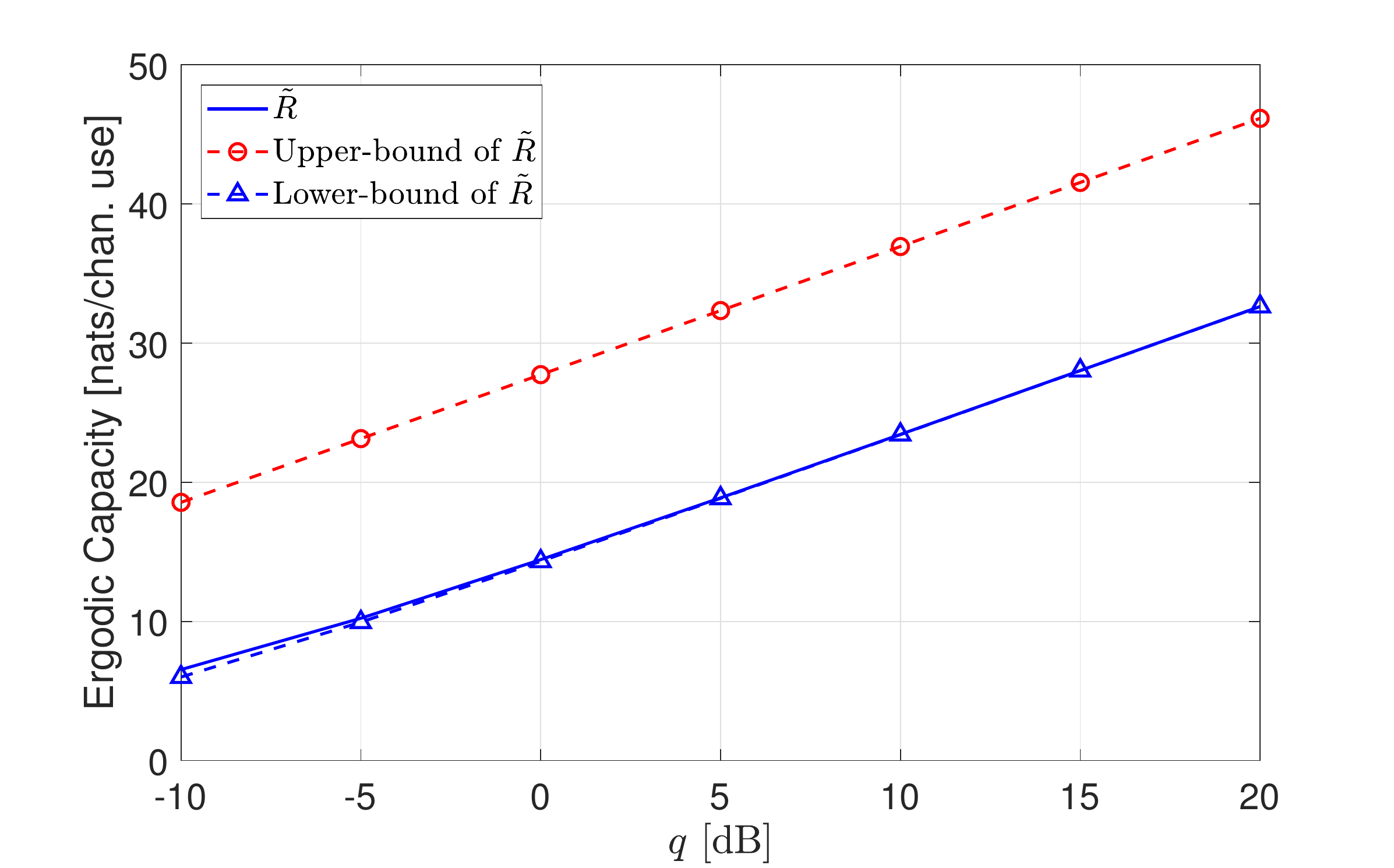}}
\vspace*{-8mm}
\caption{\label{fig4}The ergodic capacity with $K\!=\!3$, $N_0\!=\!4$, $N_1\!=\!N_2\!=\!8$, and $N_3\!=\!16$.}
\vspace*{-6mm}
\end{center}
\end{figure}

\begin{figure}[t]
\vspace*{0mm}
\begin{center}
\hspace*{-6mm}
\scalebox{0.33}{\includegraphics{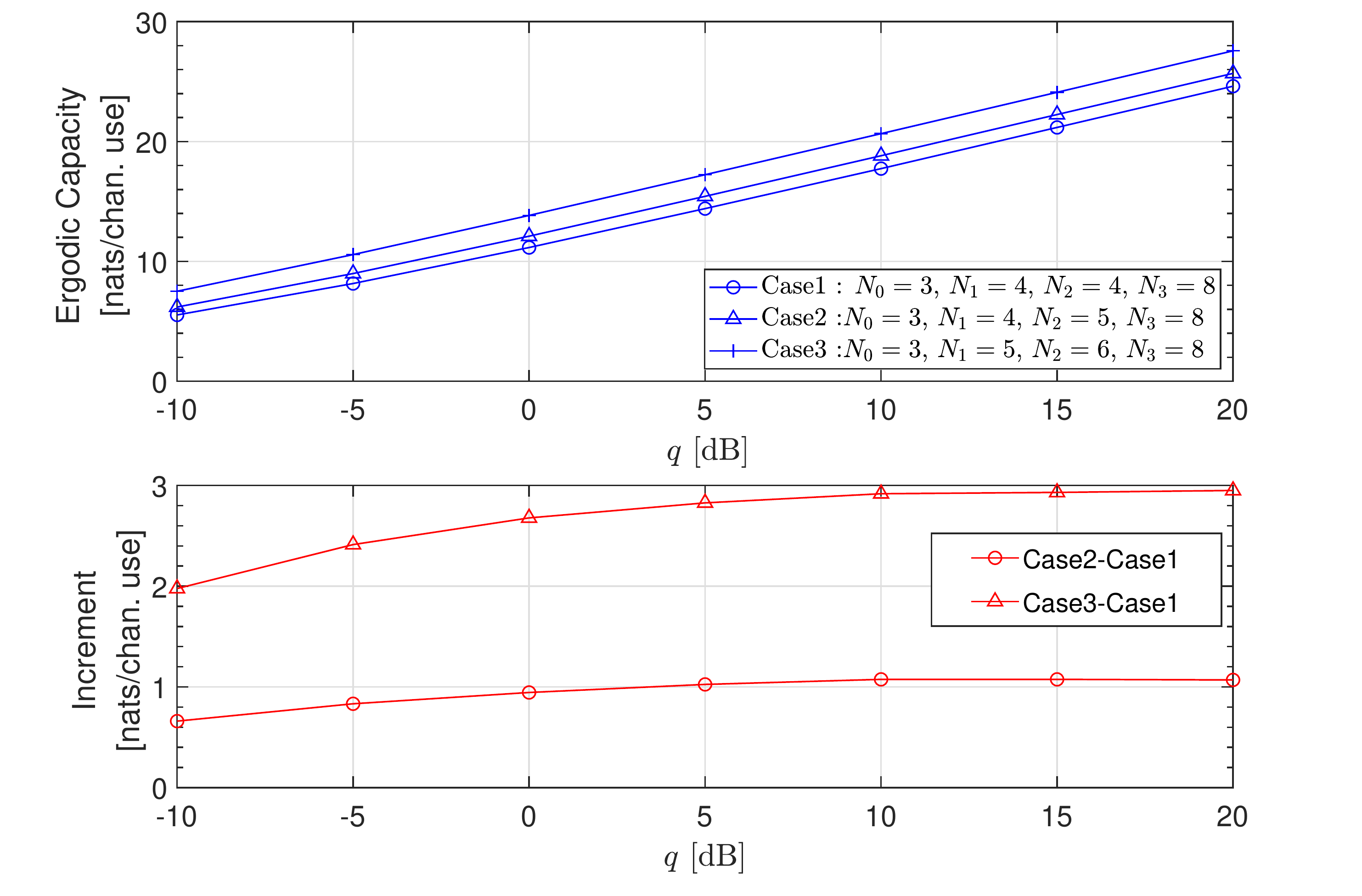}}
\vspace*{-8mm}
\caption{\label{fig5}The ergodic capacity increments with increasing the number of antennas of the UAVs.}
\vspace*{-4mm}
\end{center}
\end{figure}

\begin{figure}
\vspace*{-2mm}
\begin{center}
\hspace*{-6mm}
\scalebox{0.42}{\includegraphics{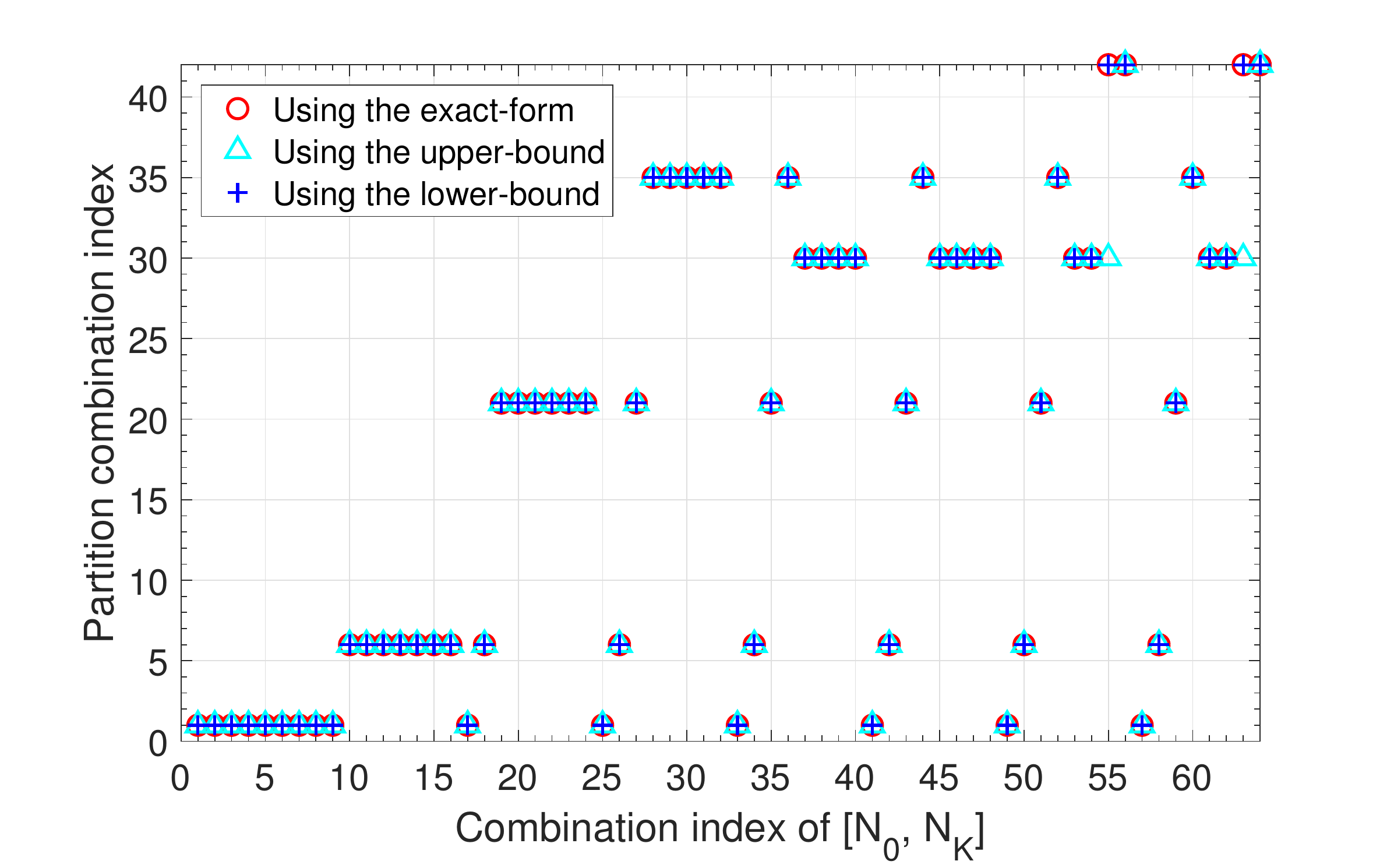}}
\vspace*{-8mm}
\caption{\label{fig6}The optimal partitions sets based on different formulas for the ergodic capacity.}
\vspace*{-6mm}
\end{center}
\end{figure}

\begin{figure}[t]
\vspace*{-2mm}
\begin{center}
\hspace*{-5mm}
\scalebox{0.42}{\includegraphics{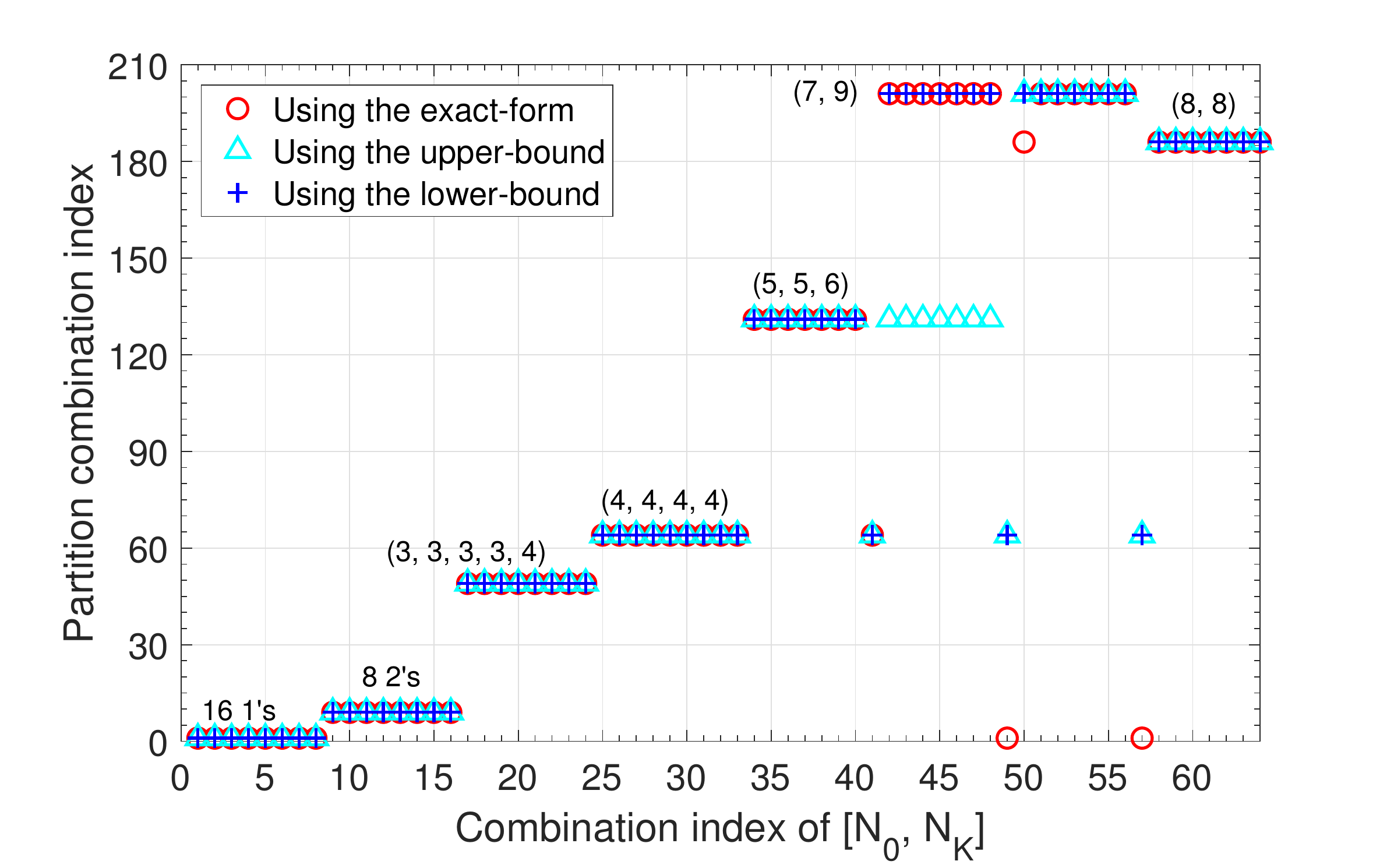}}
\vspace*{-8mm}
\caption{\label{fig7}The optimal partition sets based on different formulas for a large $M\!=\!16$.}
\vspace*{-6mm}
\end{center}
\end{figure}

\subsection{Number of Tiers Optimization}

In the remaining simulations followed, for simplicity we always assume $cp_0\!=\!1$ in (\ref{qcon}). In Fig.6, we test the number of tier optimizations using different formulas of $\tilde{R}$ with $p\!=\!20$ dB and an power-attenuation exponent $\alpha\!=\!2$. The total number of UAVs is $M\!=\!10$, and the all 42 partition sets are listed in Table~II (where the partition sets considered in the proposed scheme (\ref{subpart}) are marked in bold and italic font). We test the total 64 different combinations of $N_0$ and $N_K$, i.e., the number of antennas for the users and the BTS, respectively. For a given combination index $N$ ($1\!\leq\!N\!\leq\!64$) in the $x$-axis, the values of $N_0$ and $N_K$ are determined by
\bea \label{N0} N_0=\Bigl\lfloor \frac{N-1}{8}\Bigr\rfloor+1,\eea
and 
\bea \label{NK}  N_K\!=\!\!\!\!\!\mod(N-1,\;8)+1,\eea
respectively.

As can be seen, optimizing the number of tiers with the lower-bound yields identical outputs as the optimal case that uses the exact-form, while there are two discrepancies between the upper-bound and the exact-form based optimizations. Importantly, the optimal partition indexes only varies in 6 different indexes, 1, 6, 21, 30, 35, and 42, which correspond to the partition sets with all 1\rq{}s, all 2\rq{}s, \{3, 3, 4\}, \{5, 5\}, \{4, 6\}, and \{10\}, respectively. This is perfectly aligned with the proposed optimization routine stated in Sec. IV-B and (\ref{subpart}).

In Fig. 7, we test another case with $M\!=\!16$, $p\!=\!10$ dB, and $\alpha\!=\!3$. In this case, $\#M\!=\!231$ and we are not able to list all the partitioning sets. In this case the number of discrepancies between the upper-bound and the exact-form based optimization is 10, while that between the lower-bound and the exact-form is only 3. Moreover, as can be seen, the optimal index also only concentrated on 7 different indexes, 1, 9, 49, 64, 131, 186, and 201, which correspond to partition sets 16 1\rq{}s, 8 2\rq{}s, \{3, 3, 3, 3, 4\}, \{4, 4, 4, 4\}, \{5, 5, 6\}, \{8, 8\}, and \{7, 9\}, respectively, which are also well aligned with the proposed solution in (\ref{subpart}). This means that for the upper and lower bounds as well as exact-form based optimizations, checking the sets defined in (\ref{subpart}) provides the same results as evaluating all 231 possible sets, which has a significant complexity reduction. Furthermore, we see that the direct solution in (\ref{subpart1}) is suboptimal since the optimal partition set can be \{7, 9\} for both $N_0\!=\!6$ and 7, but (\ref{subpart1}) generates the solution for $N_0\!=\!6$ as \{6, 10\}.

In Fig. 8, we test the ergodic capacity $\tilde{R}$ with $M\!=\!8$, $N_0\!=\!4$, and $\alpha\!=\!2$. We compare the proposed solution in (\ref{subpart}), in this case, using two UAV-tiers with each containing 4 UAVs, with other heuristic schemes that use 8, 4, and 1 UAV-tiers, respectively. The values of $N_K$ changes from 8 to 64, and the power $q$ (in dB) is modeled as a random Gaussian variable with both mean and variance equal to 10. As can be seen, the partitioning set with the proposed solution quite close to the optimal scheme, with the latter one optimized individually for each SNR realization and with an exhaustive search over all 22 possible partitions. This verifies the validity of the proposed low-complexity optimization scheme with (\ref{subpart}) in practical applications.

\subsection{Asymptotic Solutions}
In Fig. 9 we test the ergodic capacity $\tilde{R}$ under different values of $p$, with $M\!=\!8$ and $\alpha\!=\!3$. We set $N_0\!=\!4$ and $N_K\!=\!8$. In this case, the number of total possible partition sets is 22. As can be seen, when $q\!=\!-15$ and 0 dB, the highest ergodic capacity is attained by a single UAV-tier, i.e., with the partition set \{8\}. When $p$ increases to 15 dB, the highest ergodic capacity is attained by partition set \{4,\,4\} with a two tiers. As $p$ further increases to 20 dB, the highest ergodic capacity now is attained by partition with all 1\rq{}s, i.e., 8 UAV-tiers and each tier contains only a single UAV. This is also well aligned with the analysis in Sec. IV-C.

\section{Summary}

In this paper, we have considered the ergodic capacity and tier optimization in UAV-assisted communication systems. With multiple tiers of UAVs, the channel between the users and the BTS can be modeled as a Rayleigh product channel. We then have derived a tight lower-bound for the ergodic capacity which is asymptotically tight in high SNR regime or with a large number of UAVs. Further, with the derived lower-bound, the ergodic capacity difference between different Rayleigh product channels can be easily computed. Furthermore, to maximize the ergodic capacity for a given total number of UAVs, we have proposed a low-complexity scheme to optimize the number of tiers in the UAV assisted systems with the derived lower-bound.

\begin{figure}[t]
\vspace*{-2mm}
\begin{center}
\hspace*{-6mm}
\scalebox{0.315}{\includegraphics{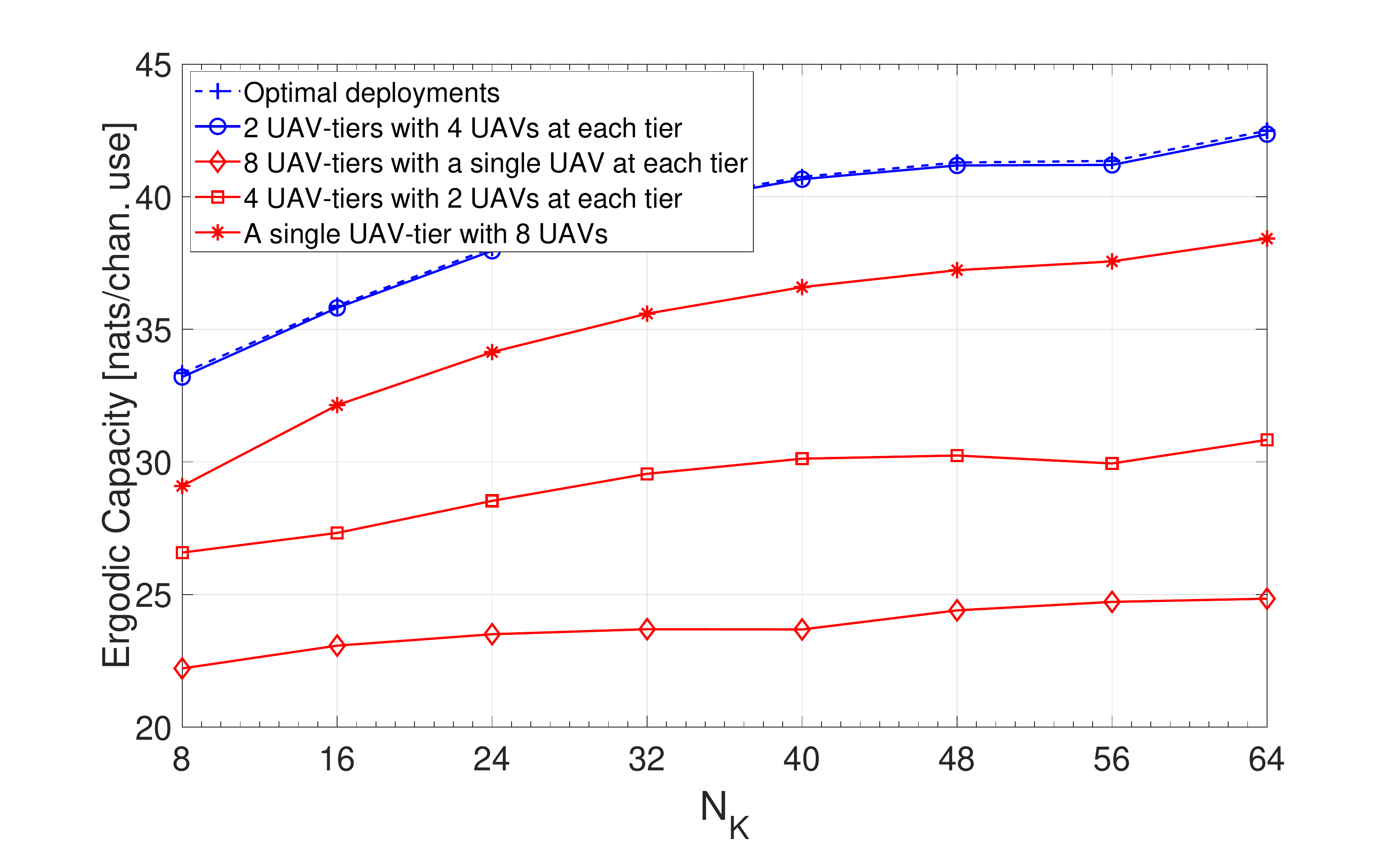}}
\vspace*{-8mm}
\caption{\label{fig8}The ergodic capacities obtained with different settings of UAV-tiers. The proposed solution in (\ref{subpart}) which is two UAV-tiers with each containing 4 UAVs, is close to the optimal scheme that is optimized individually for each SNR realization and over all possible partitions. }
\vspace*{-6mm}
\end{center}
\end{figure}

\begin{figure}[t]
\vspace*{-2mm}
\begin{center}
\hspace*{-5mm}
\scalebox{0.42}{\includegraphics{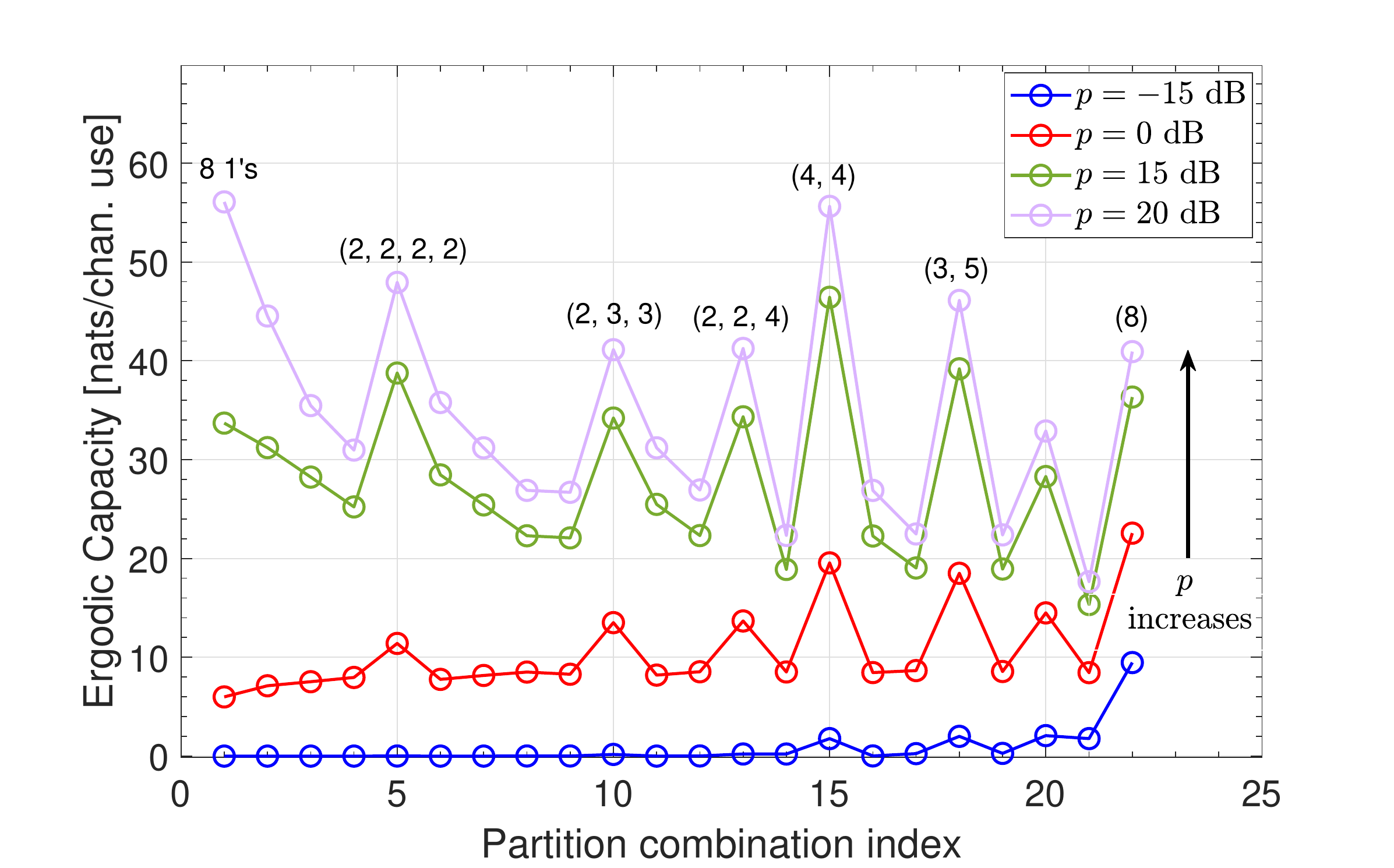}}
\vspace*{-8mm}
\caption{\label{fig9}The ergodic capacity with different partition sets for different values of $q$ with $N_0\!=\!4$, $N_K\!=\!8$, and $M\!=\!8$.}
\vspace*{-6mm}
\end{center}
\end{figure}

\begin{table*}[b]
\renewcommand{\arraystretch}{1.5}
\vspace{-0mm}
\centering
\caption{Partition Indexes and the correspondent combinations for $M\!=\!10$.}
\label{tab1}
\vspace{-1mm}
\begin{tabular}{|c|c|c|c|c|c|}
\hline
                 Index & Combination &Index& Combination & Index &Combination \\ \hhline{|=|=|=|=|=|=|} 
                 \bf{1}&\{ \bf{\emph{1     1     1     1     1     1     1     1     1     1}} \}&2    &\{ 1     1     1     1     1     1     1     1     2 \}  &3    &\{1     1     1     1     1     1     2     2 \}\\ \hline
                             4&\{1     1     1     1     2     2     2\}  &5   &\{ 1     1     2     2     2     2  \} &\bf{6}  &\{  \bf{\emph{2     2     2     2     2}}  \} \\ \hline
                               7 &\{ 1     1     1     1     1     1     1     3 \} &8 & \{ 1     1     1     1     1     2    3 \}    &9 &\{ 1     1     1    2     2    3 \} \\ \hline
  10&\{ 1     2   2     2    3 \}&11  &\{ 1     1     1     1   3    3 \}      &12  &\{ 1     1     2   3    3 \}  \\
 \hline
13& \{ 2     2   3    3 \} & 14 &\{ 1     3   3    3 \}  &15      &\{  1     1     1     1     1     1     4    \}  \\ \hline
16 &\{   1     1     1     1    2    4   \}  &17      &\{    1     1     2  2     4      \}  
&18& \{ 2     2   2    4 \}  \\ \hline   19 &\{  1     1     1     3     4   \}  &20      &\{  1     2     3     4      \}  
&\bf{21}& \{  \bf{\emph{3     3     4 }}  \} \\ \hline  22 &\{     1     1     4     4    \}  &23      &\{  2  4      4    \}  
&24& \{   1     1     1     1    5  \} \\ \hline  25 &\{ 1  1     1   2    5 \}  &26     &\{ 1 2  2    5    \}  
&27& \{ 1    1   3    5 \} \\ \hline 28 &\{ 2 3 5 \}  & 29      &\{  1 4 5   \}  
&\bf{30}& \{\bf{ \emph{5     5}} \} \\ \hline  31 &\{  1     1     1     1     6 \}  & 32      &\{  1     1    2  6    \} 
&33& \{ 2     2   6 \} \\ \hline  34 &\{ 1     3  6 \}  & \bf{35}      &\{ \bf{\emph{4 6}}   \}  
&36& \{ 1     1    1  7  \} \\ \hline 37 &\{ 1 2 7 \}  &38      &\{  3 7    \}  
&39& \{ 1 1 8\} \\ \hline  40 &\{ 2 8 \}  &41      &\{ 1 9 \} &\bf{42}      &\{ \bf{\emph{10}} \} \\  \hline
\end{tabular}
\vspace{-2mm}
\end{table*}

\section*{Appendix A: Proof of Property 1}
Arguments leading to Property 1 can be found in previous work \cite{RK14, AK13}. Here we provide a simpler and more straightforward proof. We first prove that switching any two adjacent parameters $N_i$ and $N_{i+1}$ will not change the ergodic capacity. That is, the ergodic capacities are the same with Rayleigh product channel of settings ($N_0, N_1, \cdots,N_K$) and ($N_0, N_1, \cdots,N_{i-1}, N_{i+1}, N_i, N_{i+2},\cdots,N_K$), where the parameters $N_i$ and $N_{i+1}$ ($0\!<\!i\!<\!K$) are switched. 

For the latter one, we let
\bea \label{Hmd1} \tilde{\vec{H}}\!&=&\!\left(\prod\limits_{k=1}^{i-1}\vec{Q}_k\right) \left(\tilde{\vec{Q}}_i\tilde{\vec{Q}}_{i+1}\tilde{\vec{Q}}_{i+2}\right)\left(\prod\limits_{k=i+3}^{K}\vec{Q}_k\right)\notag \\
\!&=&\!\vec{A}\left(\tilde{\vec{Q}}_i\tilde{\vec{Q}}_{i+1}\tilde{\vec{Q}}_{i+2}\right)\vec{B}, \eea
where $\vec{A}\!=\!\prod\limits_{k=1}^{i-1}\vec{Q}_k$ and $\vec{B}\!=\!\prod\limits_{k=i+3}^{K}\vec{Q}_k$. Since the statistic properties of $\vec{A}$ and $\vec{B}$ remain the same for these two different settings, it is sufficient to show that 
\bea \label{W1} \tilde{\vec{W}}=\tilde{\vec{Q}}_i\tilde{\vec{Q}}_{i+1}\tilde{\vec{Q}}_{i+2} \eea
also has the same statistic properties as
\bea \label{W2} \vec{W}=\vec{Q}_i\vec{Q}_{i+1}\vec{Q}_{i+2}. \eea
The difference between (\ref{W1}) and (\ref{W2}) is that $\tilde{\vec{Q}}_i$, $\tilde{\vec{Q}}_{i+1}$, and $\tilde{\vec{Q}}_{i+2}$ are with dimensions $N_{i-1}\!\times\!N_{i+1}$, $N_{i+1}\!\times\!N_i$, and $N_i\!\times\!N_{i+2}$; while  $\vec{Q}_i$, $\vec{Q}_{i+1}$, and $\vec{Q}_{i+2}$ are with dimensions $N_{i-1}\!\times\!N_i$, $N_i\!\times\!N_{i+1}$, and $N_{i+1}\!\times\!N_{i+2}$, respectively.

Denoting ${\vec{a}}(m,n)$ as the element on the $m$th row and $n$th column of a matrix $\vec{A}$, the element $\tilde{\vec{w}}(m,n)$ of $\tilde{\vec{W}}$ equals
\bea \label{w1} \tilde{\vec{w}}(m,n)=\!\!\sum_{s=0}^{N_{i+1}-1}\!\sum_{r=0}^{N_i-1}\tilde{\vec{q}}_i(m,s)\tilde{\vec{q}}_{i+1}(s,r)\tilde{\vec{q}}_{i+2}(r,n),\;\eea
and the element $\vec{w}(m,n)$ of $\vec{W}$ equals
\bea \label{w2} \vec{w}(m,n)=\!\!\sum_{s=0}^{N_i-1}\!\sum_{r=0}^{N_{i+1}-1}\vec{q}_i(m,s)\vec{q}_{i+1}(s,r)\vec{q}_{i+2}(r,n),\;\eea
respectively. Since $\vec{Q}_k$ and $\tilde{Q}_k$ are Rayleigh MIMO channels, elements in $\vec{q}_k$ and $\tilde{\vec{q}}_k$ are complex Gaussian variables with zero-mean and unit-variance. Therefore, $\vec{w}(m,n)$ and $\tilde{\vec{w}}(m,n)$ have the same pdf. Further, since $\vec{W}$ and $\tilde{\vec{W}}$ have the same size $N_{i-1}\!\times\!N_{i+2}$, they also have identical pdf. That is to say,  $\vec{H}$ and $\tilde{\vec{H}}$ have the same statistical properties and the ergodic capacities of them are the same.

Following similar discussions, we can also show that the above conclusion holds for the cases $i\!=\!0$ (switching $N_0$ and $N_1$) and $i\!=\!K-1$ (switching $N_{K-1}$ and $N_K$). Hence, we conclude that switching any two adjacent parameters in ($N_0, N_1, \cdots,N_K$) will not change the ergodic capacity. Since any permutation of ($N_0, N_1, \cdots,N_K$) can be decomposed as a constitution of operations that switching the order of two adjacent parameters, the Property 1 thusly holds.

One remark from the above proof is that Property 1 holds for a general condition that all elements in all matrices $\vec{Q}_k$ are i.i.d. (but not necessary Gaussian).

\section*{Appendix B: Proof of Lemma 1}
We prove Theorem 1 by deduction. Firstly, we assume the SVD decomposition
\bea \label{svd1}  \vec{Q}_1= \vec{U}\rmh\vec{\Lambda}\vec{V}\rmh, \eea
where $\vec{U}$ and $\vec{V}$ are unitary matrices with dimensions $N_1\!\times\!N_1$ and $N_0\!\times\!N_0$, respectively. The matrix $\vec{\Lambda}$ has dimensions $N_1\!\times\!N_0$ and the last $N_1\!-\!N_0$ diagonal elements are 0s. That is
\[
  \vec{\Lambda} =
  \begin{bmatrix}
    \hat{\vec{\Lambda}}  \\
    \vec{0}_{(N_1-N_0)\times N_0}\\
 
  \end{bmatrix}\!,
\]
where $\hat{\vec{\Lambda}}$ is diagonal and with dimensions $N_0\!\times\!N_0$.

Letting
$$ \vec{A}=\prod\limits_{k=3}^{K}\vec{Q}_k,$$
it holds that
\bea \label{ab2} &&\!\!\!\!\!\!\!\!\!\!\!\!\!\!\!\!\!\!\!\!\mathbb{E}\!\left[\ln\det\!\left(\!\vec{H}\rmh\vec{H}\right)\right]\!\qquad \qquad \notag \\
 &&\!\!\!\!\!\!\!\!\!\!\!\!\!\!\!\!\!\!\!=\mathbb{E}_{\{\vec{Q}_1,\,\vec{Q}_2,\,\cdots,\,\vec{Q}_K\}}\!\!\left[\ln\det\!\left(\vec{Q}_1\rmh\vec{Q}_2\rmh\vec{A}\rmh\vec{A}\vec{Q}_2\vec{Q}_1\right)\!\right]\!.  \eea
Inserting (\ref{svd1}) back into (\ref{ab2}) yields
 \bea \label{ab3} &&\!\!\!\!\!\!\!\!\mathbb{E}\!\left[\ln\det\!\left(\!\vec{H}\rmh\vec{H}\right)\right]\!\qquad \qquad \notag \\
 &&\!\!\!\!\!\!\!=\mathbb{E}_{\{\vec{\Lambda},\,\vec{V},\,\tilde{\vec{Q}}_2,\,\cdots,\,\vec{Q}_K\}}\!\!\left[\ln\det\!\left(\!\vec{V}\vec{\Lambda}\rmh\tilde{\vec{Q}}_2\rmh\vec{A}\rmh\vec{A}\tilde{\vec{Q}}_2\vec{\Lambda}\vec{V}\rmh\right)\!\right]\! \notag \\
  &&\!\!\!\!\!\!\!=\mathbb{E}_{\{\vec{\Lambda},\,\tilde{\vec{Q}}_2,\,\cdots,\,\vec{Q}_K\}}\!\!\left[\ln\det\!\left(\!\vec{\Lambda}\rmh\tilde{\vec{Q}}_2\rmh\vec{A}\rmh\vec{A}\tilde{\vec{Q}}_2\vec{\Lambda}\right)\!\right]\! \notag \\
   &&\!\!\!\!\!\!\!\overset{(a)}{=}\mathbb{E}_{\{\vec{\Lambda},\,\vec{Q}_2,\,\cdots,\,\vec{Q}_K\}}\!\!\left[\ln\det\!\left(\!\vec{\Lambda}\rmh\vec{Q}_2\rmh\vec{A}\rmh\vec{A}\vec{Q}_2\vec{\Lambda}\right)\!\right]\! \notag \\
    &&\!\!\!\!\!\!\!\overset{(b)}{=}\mathbb{E}_{\{\hat{\vec{\Lambda}},\,\hat{\vec{Q}}_2,\,\cdots,\,\vec{Q}_K\}}\!\!\left[\ln\det\!\left(\!\hat{\vec{\Lambda}}\rmh\hat{\vec{Q}}_2\rmh\vec{A}\rmh\vec{A}\hat{\vec{Q}}_2\hat{\vec{\Lambda}}\right)\!\right]\! \notag \\
        &&\!\!\!\!\!\!\!=\mathbb{E}_{\{\hat{\vec{\Lambda}}\}}\!\!\left[\ln\det\!\left(\!\hat{\vec{\Lambda}}\hat{\vec{\Lambda}}\rmh\right)\!\right]\!+\!\mathbb{E}_{\{\hat{\vec{Q}}_2,\,\cdots,\,\vec{Q}_K\}}\!\!\left[\ln\det\!\left(\!\hat{\vec{Q}}_2\rmh\vec{A}\rmh\vec{A}\hat{\vec{Q}}_2\right)\!\right]\!, \notag \\
     &&\!\!\!\!\!\!\!\overset{(c)}{=}\mathbb{E}_{\{\vec{Q}_1\}}\!\!\left[\ln\det\!\left(\!\vec{Q}_1\rmh\vec{Q}_1\right)\!\right]  \notag \\ 
     &&+\mathbb{E}_{\{\hat{\vec{Q}}_2,\,\cdots,\,\vec{Q}_K\}}\!\!\left[\ln\det\!\left(\!\hat{\vec{Q}}_2\rmh\vec{A}\rmh\vec{A}\hat{\vec{Q}}_2\right)\!\right]\!, \!\!\! \eea
where
 $$\tilde{\vec{Q}}_2=\vec{U}\rmh\vec{Q}_2,$$
and $\hat{\vec{Q}}_2$ denotes the submatrix of $\tilde{\vec{Q}}_2$ obtained by by removing the last $N_1\!-\!N_0$ columns, which is a Rayleigh channel with dimensions dimensions $N_2\!\times\!N_0$.

The equation \lq{}(a)\rq{} holds because $\tilde{\vec{Q}}_2$ and $\vec{Q}_2$ has the same statistic properties since $\vec{U}$ is unitary, and \lq{}(b)\rq{} holds since the diagonal matrix $\vec{\Lambda}$ has the last $N_1\!-\!N_0$ diagonal elements as 0s. The equation \lq{}(c)\rq{} holds due to the fact that
$$\mathbb{E}\!\left[\ln\det\!\left(\!\vec{Q}_1\rmh\vec{Q}_1\right)\!\right]\!=\mathbb{E}\!\left[\ln\det\!\left(\!\hat{\vec{\Lambda}}\hat{\vec{\Lambda}}\rmh\right)\!\right]\!.$$ 

Since the last term in (\ref{ab3}) can be equivalently rewritten as
$$\mathbb{E}_{\{\hat{\vec{Q}}_2,\,\cdots,\,\vec{Q}_K\}}\!\!\left[\ln\det\!\left(\!\hat{\vec{Q}}_2\rmh\vec{A}\rmh\vec{A}\hat{\vec{Q}}_2\right)\!\right]\!=\mathbb{E}\!\left[\ln\det\!\left(\!\hat{\vec{H}}\rmh\hat{\vec{H}}\right)\right]\!,$$
where 
$$\hat{\vec{H}}=\hat{\vec{Q}}_2\vec{A}$$
is a Rayleigh product channel with  $K\!-1\!$ components and with parameter settings $(N_0,\,N_2,\,\cdots,\,N_K)$. Following the same analysis as in (\ref{ab3}) it holds that
\bea \mathbb{E}\!\left[\ln\det\!\left(\!\hat{\vec{H}}\rmh\hat{\vec{H}}\right)\right]\!&=&\!\mathbb{E}_{\{\hat{\vec{Q}}_2\}}\!\!\left[\ln\det\!\left(\!\hat{\vec{Q}}_2\rmh\hat{\vec{Q}}_2\right)\!\right]  \notag \\ 
     &&+\mathbb{E}_{\{\hat{\vec{Q}}_3,\,\cdots,\,\vec{Q}_K\}}\!\!\left[\ln\det\!\left(\!\hat{\vec{Q}}_3\rmh\vec{B}\rmh\vec{B}\hat{\vec{Q}}_3\right)\!\right]\!, \notag \\ \eea
where $\hat{\vec{Q}}_3$ is a Rayleigh channel with dimensions $N_3\!\times\!N_0$ and
$$ \vec{B}=\prod\limits_{k=4}^{K}\vec{Q}_k.$$
Repeat such a process it can be shown that
$$\mathbb{E}\!\left[\ln\det\!\left(\!\vec{H}\rmh\vec{H}\right)\right]\!=\! \sum_{k=1}^{K}\mathbb{E}\!\left[\ln\det\!\left(\!\hat{\vec{Q}}_k\rmh\hat{\vec{Q}}_k\right)\!\right]\!, $$
where $\hat{\vec{Q}}_k$ are Rayleigh channel with dimensions $N_k\!\times\!N_0$, which proves Lemma 1.

\section*{Appendix C: Proof of Property 2}

Applying Minkowski\rq{}s inequality \cite{HJ85} for $N_0\!\times\!N_0$ positive definite matrix,
\bea \big(\det(\vec{A}+\vec{B})\big)^{1/N_0}\geq \big(\det\vec{A}\big)^{1/N_0}+\big(\det\vec{B}\big)^{1/N_0}, \notag \eea
the ergodic capacity $\tilde{R}$ in (\ref{tR}) satisfies
{\setlength\arraycolsep{2pt} \bea \label{tRab1} \tilde{R}&\geq&N_0\mathbb{E}\!\left[\ln\!\left(1+q\big(\!\det(\vec{H}\vec{H}\rmh)\big)^{1/N_0}\right)\!\right]\notag \\
&=&N_0\mathbb{E}\!\left[\ln\!\left(1+q\exp\!\left(\frac{1}{N_0}\ln\det(\vec{H}\vec{H}\rmh)\right)\!\right)\!\right].\eea}
\hspace{-1.2mm}By Jessen\rq{}s inequality, it holds from (\ref{tRab1}) that
\bea \label{tRab2} \tilde{R}\geq N_0\ln\!\left(1+q\exp\!\left(\frac{1}{N_0}\mathbb{E}\!\left[\ln\det(\vec{H}\vec{H}\rmh)\right]\right)\!\right).\eea
Noticing that with Rayleigh product channel $\vec{H}$, it holds from Theorem 1 that
\bea \label{c3} \mathbb{E}\left[ \ln\det(\vec{H}\vec{H}\rmh)\right]=N_0(g-\gamma K) ,  \eea
where $g$ is given in (\ref{gK}). Inserting (\ref{c3}) back into (\ref{tRab2}), the ergodic capacity is lower bounded as
 \bea \label{tRab3} \tilde{R}\geq N_0\ln\!\left(1+q\exp\big(g-\gamma K\big)\right)\!. \notag \eea

\section*{Appendix D: Proof of Property 3}

From Theorem 1, at high SNR the ergodic capacity difference, between the Rayleigh product channel with settings ($N_0,\,N_1,\,\cdots,N_K$) and $N_k\!=\!N_0$ for all $0\!\leq\!k\!\leq\!K$, equals
\bea \Delta\tilde{R}=\sum_{k=1}^{K}\sum_{\ell=1}^{N_{0}}\psi(N_k-\ell+1)-\sum_{k=1}^{K}\sum_{\ell=1}^{N_{0}}\psi(N_0-\ell+1) . \notag \eea
Using the identity \cite{LN13}
$$ \psi(x+1)=\frac{1}{x}+ \psi(x),$$
it holds that
\bea \Delta\tilde{R}= \sum_{k=1}^{K}\sum_{\ell=1}^{N_{0}}\sum_{s=N_0}^{N_k-1}\frac{1}{s-\ell+1}. \notag \eea

\end{document}